\documentclass[a4paper,UKenglish, thm-restate]{lipics-v2021}

\usepackage{graphicx,amsmath,amssymb,amsthm,multicol,color,makecell,environ,tabularx,mathtools,enumitem}

\usepackage[table, svgnames]{xcolor}
\usepackage{hyperref}

\newcommand{\N}{\ensuremath{\mathbb{N}}}
\renewcommand{\O}{\ensuremath{\mathcal{O}}}
\newcommand{\M}{\ensuremath{\mathcal{M}}}

\makeatletter
\newcommand{\problemtitle}[1]{\gdef\@problemtitle{#1}}
\newcommand{\probleminput}[1]{\gdef\@probleminput{#1}}
\newcommand{\problemquestion}[1]{\gdef\@problemquestion{#1}}
\newcommand{\problemparameter}[1]{\gdef\@problemparameter{#1}}
\NewEnviron{problemP}{
	\problemtitle{}\probleminput{}\problemquestion{}
	\BODY
	\par\addvspace{.5\baselineskip}
	\noindent
	\centering{
	\fbox{\parbox{.9\textwidth}{
		\@problemtitle \\
		\textit{Input: } \@probleminput \\
		\textit{Question: }	\@problemquestion
	}}}
	\par\addvspace{.5\baselineskip}
}

\NewEnviron{problemParam}{
	\problemtitle{}\probleminput{}\problemquestion{}\problemparameter{}
	\BODY
	\par\addvspace{.5\baselineskip}
	\noindent
	\centering{
	\fbox{\parbox{.9\textwidth}{
			\@problemtitle \hfill\\%
			\textit{Input: } \@probleminput \\
			\textit{Parameter: } \@problemparameter \\
			\textit{Question: }	\@problemquestion
	}}}
	\par\addvspace{.5\baselineskip}
}
\makeatother

\pdfoutput=1 
\hideLIPIcs 

\title{Parameterized Complexity of Streaming Diameter and Connectivity Problems\thanks{This version of the article has been accepted for publication, after peer review but is not the Version of Record and does not reflect post-acceptance improvements, or any corrections. The Version of Record is available online at: \url{https://dx.doi.org/10.1007/s00453-024-01246-z}. An extended abstract of this work appeared in the proceedings of IPEC 2022.}}
\author{Jelle J. Oostveen}{Dept.\ Information and Computing Sciences, Utrecht University, The Netherlands}{j.j.oostveen@uu.nl}{https://orcid.org/0009-0009-4419-3143}{This author is partially supported by the NWO grant OCENW.KLEIN.114 (PACAN).}
\author{Erik Jan {van Leeuwen}}{Dept.\ Information and Computing Sciences, Utrecht University, The Netherlands}{e.j.vanleeuwen@uu.nl}{https://orcid.org/0000-0001-5240-7257}{}

\authorrunning{J.J. Oostveen and E.J. van Leeuwen}
\titlerunning{Parameterized Complexity of Streaming Diameter and Connectivity Problems}

\Copyright{Jelle J. Oostveen and Erik Jan van Leeuwen}



\nolinenumbers

\begin{document}

\maketitle

\begin{abstract}
We initiate the investigation of the parameterized complexity of \textsc{Diameter} and \textsc{Connectivity} in the streaming paradigm. On the positive end, we show that knowing a vertex cover of size $k$ allows for algorithms in the Adjacency List (AL) streaming model whose number of passes is constant and memory is $\O(\log n)$ for any fixed $k$. Underlying these algorithms is a method to execute a breadth-first search in $\O(k)$ passes and $\O(k \log n)$ bits of memory. On the negative end, we show that many other parameters lead to lower bounds in the AL model, where $\Omega(n/p)$ bits of memory is needed for any $p$-pass algorithm even for constant parameter values. In particular, this holds for graphs with a known modulator (deletion set) of constant size to a graph that has no induced subgraph isomorphic to a fixed graph $H$, for most $H$. For some cases, we can also show one-pass, $\Omega(n \log n)$ bits of memory lower bounds. We also prove a much stronger $\Omega(n^2/p)$ lower bound for \textsc{Diameter} on bipartite graphs.

Finally, using the insights we developed into streaming parameterized graph exploration algorithms, we show a new streaming kernelization algorithm for computing a vertex cover of size~$k$. This yields a kernel of $2k$ vertices (with $\O(k^2)$ edges) produced as a stream in $\text{poly}(k)$ passes and only $\O(k \log n)$ bits of memory.
\end{abstract}

\section{Introduction}
Graph algorithms, such as to compute the diameter of an unweighted graph (\textsc{Diameter}) or to determine whether it is connected (\textsc{Connectivity}), often rely on keeping the entire graph in (random access) memory. However, very large networks might not fit in memory. Hence, graph streaming has been proposed as a paradigm where the graph is inspected through a so-called stream, in which its edges appear one by one~\cite{HenzingerStreams}. To compensate for the assumption of limited memory, multiple passes may be made over the stream and computation time is assumed to be unlimited. The complexity theory question is which problems remain solvable and which problems are hard in such a model, taking into account trade-offs between the amount of memory and passes. 

Many graph streaming problems require $\Omega(n)$ bits of memory \cite{OmegaNMemory,GraphDistancesDataStreaming} for a constant number of passes on $n$-vertex graphs. 
Any $p$-pass algorithm for \textsc{Connectivity} needs $\Omega(n/p)$ bits of memory~\cite{HenzingerStreams}. Single pass algorithms for \textsc{Connectivity} or \textsc{Diameter} need $\Omega(n \log n)$ bits of memory on sparse graphs~\cite{SunWoodruffPermBounds}.
A $2$-approximation of \textsc{Diameter} requires $\Omega(n^{3/2})$ bits of memory on weighted graphs~\cite{GraphDistancesDataStreaming}.
A naive streaming algorithm for \textsc{Connectivity} or \textsc{Diameter} stores the entire graph, using $\O(m \log n) = \O(n^2 \log n)$ bits and a single pass. For \textsc{Connectivity}, union-find yields a $1$-pass, $\O(n \log n)$ bits of memory, algorithm~\cite{McGregorSurveyStream}.

An intriguing aspect on \textsc{Diameter} and \textsc{Connectivity} is that classic algorithms for them rely on breadth-first search (BFS) or depth-first search (DFS). These seem difficult to execute efficiently in a streaming setting. It was a longstanding open problem to compute a DFS tree using $o(n)$ passes and $o(m \log n)$ bits of memory. This barrier was recently broken~\cite{KhanM19}, through an algorithm that uses $\O(n/k)$ passes and $\O(nk \log n)$ bits of memory, for any $k$.
The situation for computing single-source shortest paths seems similar~\cite{Elkin20a}, although good approximations exist even on weighted graphs (see e.g.~\cite{McGregorSurveyStream,ElkinT2021}).
We do know that DFS algorithms cannot be executed in logarithmic space~\cite{DFSsequential}. In streaming, any BFS algorithm that explores $k$ layers of the BFS tree must use at least $k/2$ passes or $\Omega(n^{1+1/k}/(\log n)^{1/k})$ space~\cite{GraphDistancesDataStreaming}.
Hence, much remains unexplored when it comes to graph exploration- and graph distance-related streaming problems such as BFS/DFS, \textsc{Diameter}, and \textsc{Connectivity}. In particular, most lower bounds hold for general graphs. As such, a more fine-grained view of the complexity of these problems has so far been lacking.

In this paper, we seek to obtain this fine-grained view using parameterized complexity~\cite{DowneyFellowsBook}.
The idea of using parameterized complexity in the streaming setting was first introduced by Fafianie and Kratsch~\cite{FafianieKratsch} and Chitnis et al.~\cite{ChitnisMaxMatchVC}. Many problems are hard in streaming parameterized by their solution size~\cite{FafianieKratsch,ChitnisMaxMatchVC,ChitnisTheory}. Crucially, however, deciding whether a graph has a vertex cover of size $k$ has a one-pass, $\O(k^2\log k \log n)$-memory kernelization algorithm by Chitnis et al.~\cite{ChitnisEsfandiariSampling}, and a $2^k$-passes, $\O(k\log n)$-memory direct algorithm by Chitnis and Cormode~\cite{ChitnisTheory}. 
Bishnu et al. \cite{BishnuStreamingVCconference} then showed that knowing a vertex cover of size $k$ is useful in solving other deletion problems using $p(k)$ passes and $f(k) \log n$ memory, notably $H$-free deletion; this approach was recently expanded on by Oostveen and van Leeuwen \cite{StreamingDeletionVC}. This leads to the more general question how knowing a (small) \emph{$H$-free modulator}, that is, a set $X$ such that $G-X$ has no induced subgraph isomorphic to $H$ (note that $H=P_2$ in \textsc{Vertex Cover [$k$]}\footnote{See Section~\ref{sec:prelims} for the notation.}), would affect the complexity of streaming problems and of BFS/DFS, \textsc{Diameter}, and \textsc{Connectivity} in particular. 
We are not aware of any investigations in this direction.

An important consideration is the streaming model 
(see~\cite{HenzingerStreams,GoelKK12,Kapralov13,McGregorVV16} or the survey by McGregor~\cite{McGregorSurveyStream}). In the Edge Arrival (EA) model, each edge of the graph appears once in the stream, and the edges appear in some fixed but arbitrary order. Most aforementioned results use this model. In the Vertex Arrival (VA) model the edges arrive grouped per vertex, and an edge is revealed at its endpoint that arrives latest.
In the Adjacency List (AL) model the edges also arrive grouped per vertex, but each edge is present for both its endpoints. This means we see each edge twice and when a vertex arrives we immediately see all its adjacencies (rather than some subset dependent on the arrival order, as in the VA model). 
This model is quite strong, but as we shall see, unavoidable for our positive results.
We do not consider dynamic streaming models in this paper, although they do exist. Note that all these streaming models concern edge streams, so even though we may talk about `arrival of a vertex' this really means `arrival of the edges incident to a vertex', and vertices are not separately present in the graph stream.

\subparagraph*{Our Contributions}
The main takeaway from our work is that the vertex cover number likely sits right at the frontier of parameters that are helpful in computing \textsc{Diameter} and \textsc{Connectivity}. As our main positive result, we show the following.

\begin{restatable}{theorem}{VCDiameterConnSummary}\label{thm:VCDiameterConnSummary}
Given a graph $G$ as an AL stream and a vertex cover of $G$ of size $k$ in memory, \textsc{Diameter [$k$]} and \textsc{Connectivity [$k$]} can be solved using $\O(2^kk)$ passes and $\O(k \log n)$ bits of space or using one pass and $\O(2^k + k \log n)$ bits of space.
\end{restatable}
The crux to our approach is to perform a BFS in an efficient manner, using $\O(k)$ passes and $\O(k \log n)$ space.
Knowledge of a vertex cover is not a restricting assumption, as one may be computed using similar memory requirements~\cite{ChitnisEsfandiariSampling,ChitnisTheory}. We will also show how to extend the single-pass result to work without a vertex cover being given, at the cost of increasing the memory use to $\O(4^k + k \log n)$ bits of space. 

As a contrasting result, we will show that in the VA model, even a constant-size vertex cover does not help in computing \textsc{Diameter} and \textsc{Connectivity}. Moreover, the bound on the vertex cover seems necessary, as we can prove that any $p$-pass algorithm for \textsc{Diameter} requires $\Omega(n^2/p)$ bits of memory even on bipartite graphs and any $p$-pass algorithm for \textsc{Connectivity} requires $\Omega(n/p)$ bits of memory, both in the AL model. This indicates that both the permissive AL model and a low vertex cover number are truly needed.

In some cases, we are also able to prove that a single-pass algorithm requires $\Omega(n \log n)$ bits of memory.

More broadly, knowledge of being \emph{$H$-free} (that is, not having a fixed graph $H$ as an induced subgraph) or having an $H$-free modulator does not help even in the AL model. Here, $H\not\subseteq_i G$ denotes that $H$ is not an induced subgraph of $G$.

\begin{restatable}{theorem}{HfreeOverview}\label{thm:HfreeOverview}
For any fixed graph $H$ with $H \not\subseteq_i P_4$ and $H \not= 3P_1, P_3+P_1,P_2+2P_1$, any streaming algorithm for \textsc{Diameter} in the AL model that uses $p$ passes over the stream must use $\Omega(n/p)$ bits of memory even on the class of $H$-free graphs.
\end{restatable}
We note that these results hold for $H$-free graphs (without the need for a modulator). The case when $H \subseteq_i P_4$ is straightforward to solve with $\O(\log n)$ bits of memory,  as the diameter is either $1$ or $2$ (an induced path of length~$3$ is a $P_4$). If the graph has diameter~$1$, it is a clique. This can be tested in a single pass by counting the number of edges. 

\begin{restatable}{theorem}{DistHfreeOverview}\label{thm:DistHfreeOverview}
For any fixed graph $H$ with $H \not= P_2 + sP_1$ for $s \in\{0,1,2\}$ and $H \not= sP_1$ for $s \in \{1,2,3\}$ and any fixed constant $k \geq 3$, any streaming algorithm for \textsc{Diameter} in the AL model that uses $p$ passes over the stream must use $\Omega(n/p)$ bits of memory even on the class of graphs $G$ with a given set of vertices $X$ with $|X| = k$ such that $G-X$ is $H$-free. If $G-X$ must be connected and $H$-free, then additionally $H \not= P_3$.
\end{restatable}
We note that the case when $H = P_2$ or $H=P_1$ is covered by Theorem~\ref{thm:VCDiameterConnSummary}.
Cobipartite graphs seem to be a bottleneck class. The cases when $H = 2P_1$ or when $H = P_3$ and $G-X$ must be connected lead to a surprising second positive result. 

\begin{restatable}{theorem}{DistlCliqueSummary}\label{thm:DistlCliqueSummary}
Given a graph $G = (V,E)$ as an AL stream and a set $X\subseteq V$ of size $k$ in memory such that $G-X$ is a disjoint union of $\ell$ cliques, \textsc{Diameter [$k, \ell$]} and \textsc{Connectivity [$k,\ell$]} can be solved using $\O(2^kk\ell)$ passes and $\O((k+\ell)\log n)$ bits of space or one pass and $\O(2^k\ell + (k + \ell) \log n)$ bits of space.
\end{restatable}

The approach for this result is similar as for Theorem~\ref{thm:VCDiameterConnSummary}. Moreover, we show a complementary lower bound in the VA model, even for $\ell=1$ and constant $k$.

To summarize our results in words, generalizing Theorem~\ref{thm:VCDiameterConnSummary} using the perspective of an $H$-free modulator does not seem to lead to a positive result (Theorem~\ref{thm:DistHfreeOverview}). Instead, connectivity of the remaining graph after removing the modulator seems crucial. However, this perspective only helps for Theorem~\ref{thm:DistlCliqueSummary}, while the problem remains hard for most other $H$-free modulators and even for the seemingly simple case of a modulator to a path (we will show this in Theorem~\ref{thm:DiamondResults}). While Theorem~\ref{thm:DistlCliqueSummary} would also hint at the possibility of using a modulator to a few components of small diameter, this also leads to hardness (we show this in Corollary~\ref{cor:SimpleALResults}).

We emphasize that all instances of \textsc{Diameter} in our hardness reductions are connected graphs. Hence, the hardness of computing \textsc{Diameter} is separated from the hardness of computing \textsc{Connectivity}.

For \textsc{Connectivity}, we also give two broad theorems that knowledge of being $H$-free or having an $H$-free modulator does not help even in the AL model.

\begin{restatable}{theorem}{HfreeOverviewConn}\label{thm:HfreeOverviewConn}
For any fixed graph $H$ that is not a linear forest containing only paths of length at most~$5$, any streaming algorithm for \textsc{Connectivity} in the AL model that uses $p$ passes over the stream must use $\Omega(n/p)$ bits of memory even on the class of $H$-free graphs.
\end{restatable}

\begin{restatable}{theorem}{DistHfreeOverviewConn}\label{thm:DistHfreeOverviewConn}
For any fixed graph $H$ that is not a linear forest containing only paths of length at most~$1$ and any fixed constant $k\geq 2$, any streaming algorithm for \textsc{Connectivity} in the AL model that uses $p$ passes over the stream must use $\Omega(n/p)$ bits of memory even on the class of graphs $G$ with a given set of vertices $X$ with $|X| = k$ such that $G-X$ is $H$-free.
\end{restatable}

Our hardness results for $H$-free modulators for both \textsc{Diameter} and \textsc{Connectivity} have meaning for several standard graph parameters, we explicitly introduce and mention these parameters in the main body of the paper.

As a final result, we use our insights into graph exploration on graphs of bounded vertex cover to show a result on the \textsc{Vertex Cover} problem itself. In particular, a kernel on $2k$ vertices for \textsc{Vertex Cover [$k$]} can be obtained as a stream in $\O(k^3)$ passes in the EA model using only $\O(k\log n)$ bits of memory. In the AL model, the number of passes is only $\O(k^2)$. This kernel still may have $\O(k^2)$ edges, which means that saving it in memory would not improve over the result of Chitnis et al.~\cite{ChitnisEsfandiariSampling} (which uses $\O(k^2 \log k \log n)$ bits of memory), up to logarithmic factors and constants. Indeed, a better kernel seems unlikely to exist~\cite{DellMelkebeekKernelLowerBounds}. However, the important point is that storing the (partial) kernel in memory is not needed during its computation. Hence, it may be viewed as a possible first step towards a streaming algorithm for \textsc{Vertex Cover [$k$]} using $\O(k\log n)$ bits of memory and $\text{poly}(k)$ passes, which is an important open problem in the field, see \cite{ChitnisTheory}.
Our kernel is constructed through a kernel by Buss and Goldsmith~\cite{BussVCkernel}, and then finding a maximum matching in an auxiliary bipartite graph (following Chen et al.~\cite{CHEN2001_VC}) of bounded size through repeated DFS applications.

\subparagraph*{Related work}
There has been substantial work on the complexity of graph-distance and reachability problems in the streaming setting. For example, Guruswami and Onak~\cite{GuruswamiO16} showed that any $p$-pass algorithm needs $n^{1+\Omega(1/p)}/p^{\O(1)}$ memory when given vertices $s,t$ to test if $s,t$ are at distance at most $2p+2$ in undirected graphs or to test $s$-$t$ reachability in directed graphs. Further work on directed $s$-$t$ reachability~\cite{AssadiR20} recently led to a lower bound that any $o(\sqrt{\log n})$-pass algorithm needs $n^{2-o(1)}$ bits of memory~\cite{CKPSSY21}. Other recent work considers $p$-pass algorithms for $\epsilon$-property testing of connectivity~\cite{YuVerbinCycleCounting,HuangP19,AssadiKSY20}, including strong memory lower bounds $n^{1-\O(\epsilon\cdot p)}$ on bounded-degree planar graphs~\cite{AssadiN21}. Further problems in graph streaming are extensively discussed and referenced in these works; see also~\cite{AssadiCK19-arxiv}.

In the non-streaming setting, the \textsc{Diameter} problem can be solved in $\O(nm)$ time by BFS. There is a lower bound of $n^{2-\epsilon}$ for any $\epsilon > 0$ under the Strong Exponential Time Hypothesis (SETH)~\cite{RodittyWilliamsSeth}. Parameterizations of \textsc{Diameter} have been studied with parameter vertex cover~\cite{BringmannHusfeldtDiameterTreewidth}, treewidth~\cite{AbboudWilliamsWangDiameter,HusfeldtDiameterTreewidth,BringmannHusfeldtDiameterTreewidth}, and other parameters~\cite{CoudertDucoffePopaDiameter,DiameterComplexityNonStreaming}, leading to a $2^{\O(k)} n^{1+\epsilon}$ time algorithm on graphs of treewidth~$k$ \cite{BringmannHusfeldtDiameterTreewidth}. Running time $2^{o(k)}n^{2-\epsilon}$ for graphs of treewidth $k$ is not possible under SETH~\cite{AbboudWilliamsWangDiameter}. Subquadratic algorithms are known for various hereditary graph classes; see e.g.~\cite{Cabello19,CorneilDHP01,Ducoffe21,DucoffeD21,DucoffeHV19,DucoffeHV20,GawrychowskiKMS21} and references in~\cite{CorneilDHP01}.


\section{Preliminaries}\label{sec:prelims}
We work on undirected, unweighted graphs. 
We denote a computational problem $A$ with \textsc{A~[$k$]}, where $[.]$ denotes the parameterization. A parameter $k$ is an integer given as additional input. In parameterized complexity the aim is to find algorithms with running time $f(k)\cdot n^{\O(1)}$, where $f$ is some computable function. This notion was introduced by Downey and Fellows~\cite{DowneyFellowsBook}, and we refer the reader to~\cite{BookParamComplexity} for more on parameterized complexity. In our setting, our aim will be to find streaming algorithms with a space complexity of $O(f(k) \log n)$ for some computable function $f$. This space complexity class is dubbed Fixed Parameter Streaming by Chitnis and Cormode~\cite{ChitnisTheory}.

\textsc{Diameter} is to compute $\max_{s,t\in V}d(s,t)$ where $d(s,t)$ denotes the distance between $s$ and $t$. \textsc{Connectivity} asks to decide whether or not the graph is connected.
A \emph{twin class} consists of all vertices with the same open neighbourhood. In a graph with vertex cover size $k$, we have $\O(2^k)$ twin classes. For two graphs $G,H$, $G+H$ denotes their disjoint union. We also use $2G$ to denote $G+G$; $3G$ is $G+G+G$, etc. $H\not\subseteq_i G$ denotes that $H$ is not an induced subgraph of $G$. A \emph{linear forest} is a disjoint union of paths. A path on $a$ vertices is denoted $P_a$ and has length~$a-1$.

\begin{problemP}
	\problemtitle{\textsc{Disj}$_n$ (Disjointness)}
	\probleminput{Alice has a string $x \in {\{0,1\}}^n$ given by $x_1x_2\ldots x_n$. Bob has a string $y \in {\{0,1\}}^n$ given by $y_1y_2\ldots y_n$.}
	\problemquestion{Bob wants to check if $\exists 1\leq i \leq n$ such that $x_i = y_i = 1$. (Formally, the answer is NO if this is the case.)}
\end{problemP}

The communication complexity necessary between Alice and Bob to solve \textsc{Disjointness} is well understood, and can be used to prove lower bounds on the memory use of streaming algorithms. This was first done by Henzinger et al.~\cite{HenzingerStreams}.
	
\begin{problemP}
	\problemtitle{\textsc{Perm}$_n$ (Permutation)}
	\probleminput{Alice gets a permutation $\pi : [n] \rightarrow [n]$, Bob gets an integer $j \in [n \log n]$.}
	\problemquestion{Is the $j$-th bit of the concatenation of $\pi(1), \pi(2), \ldots, \pi(n)$ equal to 1? That is, is the $\gamma = (j \bmod \log n)$-th bit of the image of the $\psi = \lceil \frac{j}{\log n} \rceil$-th index equal to 1?}
\end{problemP}
	
The values $\gamma$ and $\psi$ essentially break $j$ down into the quotient ($\psi$) and the remainder ($\gamma$) with divisor $\log n$. In essence, $\psi$ tells us in which `block' of $\log n$ bits Bob should look, while $\gamma$ tells us where in a block Bob should look. Let it be clear that $\psi \in [n]$ and $\gamma \in [\log n]$. For computation, one can compute $\gamma = (j + \log n - \lceil \frac{j}{\log n} \rceil \times \log n)$.
The problem of \textsc{Permutation} was created and applied first to the streaming setting by Sun and Woodruff~\cite{SunWoodruffPermBounds}.
The following formulations by Bishnu et al.~\cite{BishnuStreamingVC} come in very useful.

\begin{proposition}{(Rephrasing of items $(ii)$ of \cite[Proposition~5.6]{BishnuStreamingVC})}\label{prop:DisjReduction}
If we can show a reduction from \textsc{Disj}$_n$ to problem $\Pi$ in streaming model $\M$ such that in the reduction, Alice and Bob construct one model-$\M$ pass for a streaming algorithm for $\Pi$ by communicating the memory state of the algorithm only a constant number of times to each other, then any streaming algorithm working in the model $\M$ for $\Pi$ that uses $p$ passes requires $\Omega(n/p)$ bits of memory, for any $p \in \N$ \cite{ChitnisAnnouncement,BishnuFundamentalGeometric,AgarwalSpatialScan}.
\end{proposition}

\begin{proposition}{(Rephrasing of item $(iii)$ of \cite[Proposition~5.6]{BishnuStreamingVC})}\label{prop:PermReduction}
If we can show a reduction from \textsc{Perm}$_n$ to a problem $\Pi$ in the streaming model $\M$ such that in the reduction, Alice and Bob construct one model-$\M$ pass for a streaming algorithm for $\Pi$ by communicating the memory state of the algorithm only a constant number of times to each other, then any streaming algorithm working in the model $\M$ for $\Pi$ that uses $1$ pass requires $\Omega(n\log n)$ bits of memory.
\end{proposition}
If we can show a reduction from either \textsc{Disjointness} or \textsc{Permutation}, we call a problem `hard', as it does not admit algorithms using only poly-logarithmic memory.

Any upper bound for the EA model holds for all models, and an upper bound for the VA model also holds for the AL model. On the other hand, a lower bound in the AL model holds for all models, and a lower bound for the VA model also holds for the EA model.


\section{Upper Bounds for Diameter}\label{sec:upperboundsdiameter}

We give an overview of our upper bound results for \textsc{Diameter} in Table~\ref{table:DiameterUpperBounds}.
The memory-efficient results rely on executing a BFS on the graph, which is made possible by both the parameter and the use of the AL model. The one-pass results rely on the possibility to save the entire graph in a bounded fashion. Our upper bounds assume the deletion set related to the parameter is given, that is, it is in memory.

\begin{table}[ht]
\begin{tabular}{|l|l|l|l|l|}
\hline
 Parameter ($k$) & Passes & Memory (bits) & Model & Theorem\\ \hline
 \textsc{Vertex Cover} & \cellcolor{LightGreen}$\O(2^kk)$ & \cellcolor{LightGreen}$\O(k \log n)$ & \cellcolor{LightGreen}AL & \ref{thm:DiameterVC}\\ 
 & \cellcolor{LightGreen}1 & \cellcolor{LightGreen}$\O(2^k + k \log n)$ & \cellcolor{LightGreen}AL & \ref{thm:VC1pass} \\ \hline
 \textsc{Distance to $\ell$ cliques} & \cellcolor{LightGreen}$\O(2^k\ell k)$ & \cellcolor{LightGreen}$\O((k + \ell) \log n)$ & \cellcolor{LightGreen}AL & \ref{thm:DiameterDistCliques}\\
 & \cellcolor{LightGreen}1 & \cellcolor{LightGreen}$\O(2^k\ell + (k + \ell) \log n)$ & \cellcolor{LightGreen}AL & \ref{thm:Clique1pass}\\ \hline
\end{tabular}
\caption{Overview of the algorithms and their complexity for \textsc{Diameter} and \textsc{Connectivity}.}
\label{table:DiameterUpperBounds}
\end{table}

\begin{restatable}{lemma}{VCpathlength}\label{lemma:VCpathlength}
    In a graph with vertex cover size $k$, any simple path has length at most~$2k$.
\end{restatable}
\begin{proof}
    Let $G$ be a graph with vertex cover size $k$. Consider some simple path $P$ in the graph. Any vertex in the independent set of $G$ (i.e. not in the vertex cover) that is on the path $P$ only has neighbours in the vertex cover. Hence, for each vertex in the independent set the path visits, we also visit a vertex in the vertex cover. As the vertex cover has size $k$, any simple path can visit at most $2k + 1$ vertices, as then all vertices in the vertex cover have been visited.
\end{proof}

\noindent Lemma~\ref{lemma:VCpathlength} is useful in that the diameter of such a graph can be at most $2k$ if the graph is connected.
Our algorithm will simulate a BFS for $2k$ rounds, deciding on the distance of a vertex to all other vertices.

\begin{lemma}\label{lemma:BFSinVCgraph}
    Given a graph $G$ as an AL stream with a vertex cover $X$ of size $k$, we can compute the distance from a vertex $v$ to all others using $\O(k)$ passes and $\O(k \log n)$ bits of memory.
\end{lemma}
\begin{proof}
    We simulate a BFS originating at $v$ for at most $2k$ rounds on our graph, using a pass for each round. Contrary to a normal BFS, we only remember whether we visited the vertices in $X$ and their distances, to reduce memory complexity.
    
    For every vertex $w\in X$, we save its tentative distance $d(w)$ from $v$; if this is not yet decided, this field has value $\infty$. Our claim will be that after round $i$, the value of $d(w)$ for vertices $w$ within distance $i$ from $v$ is correct. We initialize the distance of $v$ as $d(v) = 0$ (we store $d(v)$ regardless of whether $v\in X$).
    
    Say we are in round $i\geq 1$. We execute a pass over the stream. Say we view a vertex $w \in X \cup \{v\}$ in the stream with its adjacencies. If $w$ has a distance of $d(w)$, we update the neighbours $u$ of $w$ in $X$ to have distance $d(u) = \min(d(u), d(w)+1)$. If instead we view a vertex $w \notin X \cup \{v\}$ in the stream, we do the following. Locally save all the neighbours and look at their distances, and let $z$ be the neighbour with minimum $d(z)$ value. For every $u\in N(w)$ we update the distance as $d(u) = \min(d(z)+2, d(u))$. This simulates the distance of a path passing through $w$ (note that this may not be the shortest path, but this may be resolved by other vertices). Executing this procedure for every vertex of $G$ takes a single pass, as by the AL model we see all the adjacencies of a vertex when it arrives in the stream. This completes the procedure for round $i$.
    
    Notice that we use only $\O(k \log n)$ bits of memory during the procedure, and that the total number of passes is indeed $\O(k)$ as we execute $2k$ rounds, using one pass each.
    
    For the correctness, let us first argue the correctness of the claim \emph{after round $i$, the value of $d(w)$ of vertices $w\in X$ within distance $i$ from $v$ is correct}. We proceed by induction, clearly the base case of 0 is correct. Now consider some vertex $w$ at distance $i$ from $v$. Consider a shortest path from $v$ to $w$. Look at the last vertex on the path before visiting $w$. If this vertex is in $X$, then by induction, this vertex has a correct distance after round $i-1$, and so, in round $i$ this vertex will update the distance of $w$ to be $i$. If this vertex is not in $X$, then it has a neighbour with distance $i-2$, which is correct after round $i-2$ by induction, and so, the vertex not in $X$ will (have) update(d) the distance of $w$ to be $i$ in round $i$.
    
    The correctness of the algorithm now follows from the claim, together with Lemma~\ref{lemma:VCpathlength}, and the fact that we can now output all distances using a single pass by either outputting the value of the field $d(w)$ for a vertex $w\in X$, or by looking at all neighbours of a vertex $w \notin X$ and outputting the smallest value $+1$.
\end{proof}

Related is a lower bound result by Feigenbaum et al.~\cite{GraphDistancesDataStreaming}, which says that any BFS procedure that explores $k$ layers of the BFS tree must use at least $k/2$ passes or super-linear memory. This indicates that memory- and pass-efficient implementations of BFS, as in Lemma~\ref{lemma:BFSinVCgraph}, are hard to come by.

We can now use Lemma~\ref{lemma:BFSinVCgraph} to construct an algorithm for finding the diameter of a graph parameterized by vertex cover, essentially by executing Lemma~\ref{lemma:BFSinVCgraph} for every twin class, which considers all options for vertices in the graph.

\begin{restatable}{theorem}{DiameterVC}\label{thm:DiameterVC}
    Given a graph $G$ as an AL stream with vertex cover $X$ of size $k$, we can solve \textsc{Diameter [$k$]} in $\O(2^kk)$ passes and $\O(k \log n)$ bits of memory.
\end{restatable}
\begin{proof}
    We enumerate all the twin classes of the neighbourhood in $X$ of vertices (which we can do with $k$ bits), and for each such a class, we find if there is a vertex realizing this class in a pass. Then, we call the algorithm of Lemma~\ref{lemma:BFSinVCgraph} with this vertex as $v$. Instead of outputting all distances, we are only interested in the largest distance found (which may be $+\infty$). We also call Lemma~\ref{lemma:BFSinVCgraph} for every $x\in X$ with $x$ as the vertex $v$. We keep track of the largest distance found over all calls to Lemma~\ref{lemma:BFSinVCgraph}, and output this value as the diameter.
    
    The correctness follows from the correctness of Lemma~\ref{lemma:BFSinVCgraph}, together with the fact that considering each twin class of the neighbourhood in $X$ combined with all vertices in $X$ actually considers all possible vertices that may occur in $G$, and so we also consider one of the vertices of the diametric pair in one of these iterations.
\end{proof}

We show an alternative one-pass algorithm, by saving the graph as a representation by its twin classes, thereby completing the proof of Theorem~\ref{thm:VCDiameterConnSummary}.

\begin{restatable}{theorem}{VConepass}\label{thm:VC1pass}
    Given a graph $G$ as an AL stream, we can solve \textsc{Diameter [$k$]} in one pass and $\O(4^k + k \log n)$ bits of memory, or correctly report that a vertex cover of size $k$ does not exist. When a vertex cover of size $k$ is given, the memory use is $\O(2^k + k \log n)$.
\end{restatable}
\begin{proof}
    In our pass, we greedily construct a vertex cover of size $2k$ by maintaining a maximal matching. If at any point the matching exceeds $2k$ vertices, we report that no vertex cover of size $k$ exists. We can characterize vertices not in the vertex cover by their adjacencies towards the vertex cover, i.e.\ the binary string of at most $2k$ bits with a 1 if the vertex is adjacent. We call this binary string the characterization of a vertex. In the pass, we also keep track of the adjacency matrix of edges within the vertex cover, and the characterization of vertices not in the vertex cover. Seeing a vertex, we either add it to the vertex cover if it has a neighbour that is not in the vertex cover (add that edge to the matching), in which case we can update the edges within the vertex cover. Otherwise, a vertex has only neighbours in the vertex cover, which means we can save its characterization. Any edge in the vertex cover will be registered, as when the second of its two vertices is added to the vertex cover, we will register the presence of this edge. Any edge with one endpoint $v$ not in the vertex cover will be registered when saving the characterization of $v$, which, at that point in the stream, can only have neighbours in the vertex cover (otherwise it would have been added too).
    
    There are only $\O(2^{2k})$ different characterizations of adjacencies to the vertex cover of size $2k$, and hence, for each we can save one or two bits whether there is a vertex with this neighbourhood and whether there is more than one. The procedure above can decide such properties locally using $\O(k \log n)$ bits. The adjacency matrix of the vertex cover takes $\O(k^2)$ bits. So, we can save all this information and the vertex cover itself using $\O(4^k + k \log n)$ bits. When a vertex cover is given, there are only $\O(2^{k})$ different characterizations and so we use only $\O(2^k + k \log n)$ bits.
    
    Next we argue that this information is enough to decide on \textsc{Diameter}. We can use a simple enumeration technique to find the diameter of the graph. To do this, for every pair, we find the distance between them, and keep track of the largest distance found. For a given pair of vertices (given by their adjacencies towards the vertex cover, or a vertex in the vertex cover itself), we can decide on the distance between them using a procedure similar to Lemma~\ref{lemma:BFSinVCgraph} but internally instead of making actual passes over the stream.
\end{proof}

We have seen all the elements of Theorem~\ref{thm:VCDiameterConnSummary}.
\begin{proof}[Proof of Theorem~\ref{thm:VCDiameterConnSummary}]
	This follows immediately from Theorem~\ref{thm:DiameterVC} and \ref{thm:VC1pass}.
\end{proof}

Next, we show that the idea of simulating a BFS extends to another similar setting, where instead of a bounded vertex cover we have a bounded deletion distance to $\ell$ cliques. The good thing about cliques is that we need not search in them, the distances in a clique are known if we know the smallest distance to some vertex in the clique. However, we will need to save the smallest distance to each clique to propagate distances in the network as different vertices in a single clique can have many different adjacencies to the deletion set. This is the reason we require a bounded number of cliques.

\begin{restatable}{lemma}{DistLCliquepathlength}\label{lemma:DistLCliquepathlength}
    In a graph $G$ with deletion distance $k$ to $\ell$ cliques, any shortest path between two vertices is of length at most $3k + 1$, if it exists.
\end{restatable}
\begin{proof}
    Let $G$ be a graph with deletion distance $k$ to $\ell$ cliques. Consider some shortest path between two vertices $v,w$. Any vertex on this path that is not $v,w$ in one of the $\ell$ cliques must have as one of its neighbours on the path a vertex in the deletion set. If the path contains more than one edge from a single clique, it is not a shortest path. Hence, the path has length at most $3k + 1$.
\end{proof}

Lemma~\ref{lemma:DistLCliquepathlength} indicates that we can use a similar approach in simulating a BFS of bounded depth to find distances.

\begin{restatable}{lemma}{BFSinDistCliquegraph}\label{lemma:BFSinDistCliquegraph}
    Given a graph $G$ as an AL stream with deletion distance $k$ to $\ell$ cliques, with the given deletion set $X$, we can decide on the distance from one vertex $v$ to all others using $\O(k)$ passes and $\O((k+\ell) \log n)$ bits of memory.
\end{restatable}
\begin{proof}
    Similar to Lemma~\ref{lemma:BFSinVCgraph}, we simulate a BFS originating at $v$ of at most $3k+1$ rounds, using a pass for each round. We only remember distances for vertices in the deletion set, and the smallest distance in each clique. This way the memory complexity remains small.
    
    The setup of the algorithm is as follows. For every vertex in the deletion set $X$, and for every clique, we save its distance from $v$, if this is not yet decided this field has value $+\infty$. For a clique, this value means the smallest distance from $v$ to \emph{some} vertex of the clique. Let us denote $d(w)$ as the value of this field for a vertex $w\in X$ or a clique with which we associate a (non-existent) vertex $w$. Our claim is that after round $i$, the fields $d(w)$ for $w$ within distance $i$ from $v$ are correct. We initialize $d(v) = 0$, and if it is contained in a clique, then set $d(w) = 0$ for the associated vertex $w$.
    
    Let us describe the workings of a round, say round $i\geq 1$. We make a pass over the stream, and for each vertex we do the following. If we see a vertex $w\in X$ with distance $d(w)$, we update the distances of its neighbours $u\in N(w)$ as $\min(d(u), d(w)+1)$. If we see a vertex $c$ contained in some clique with associated vertex $w$, we look at its neighbours in $N(c) \cap X$. If a neighbour $u$ of $c$ has $d(u) + 1 \leq d(w)$, we update $d(w)$, as $c$ realizes this distance in the clique. Therefore, we update the neighbours $u \in N(c) \cap X$ with $d(u) = \min(d(u), d(w)+1)$. Otherwise, $d(w)$ is not realized by $c$, but by another vertex of the clique, and so, we can update the neighbours $u$ of $c$ in $X$ with $d(u) = \min(d(u), d(w) + 2)$. This concludes what we do in a round.
    
    Notice that we can always identify which clique a vertex not in $X$ belongs to, as the AL stream provides all its neighbours.
    
    The correctness of this algorithm quickly follows from the correctness of Lemma~\ref{lemma:BFSinVCgraph} combined with how we handle cliques here, and Lemma~\ref{lemma:DistLCliquepathlength}. Let it also be clear that we use $\O(k)$ passes and use $\O((k+\ell)\log n)$ bits of memory.
\end{proof}

Using Lemma~\ref{lemma:BFSinDistCliquegraph} we can decide the diameter of the graph by calling it many times for the possible vertices in the graph.

\begin{restatable}{theorem}{DiameterDistCliques}\label{thm:DiameterDistCliques}
    Given a graph $G$ as an AL stream with deletion distance $k$ to $\ell$ cliques, with the given deletion set $X$, we can solve \textsc{Diameter [$k,\ell$]} using $\O(2^k\ell k)$ passes and $\O((k+\ell) \log n)$ bits of memory.
\end{restatable}
\begin{proof}
    We can enumerate all the twin classes of neighbourhoods in $X$ (with $k$ bits), and for each class use a pass to find at most $\ell$ vertices which realize this neighbourhood in $X$ (at most one from each clique, two vertices from the same clique with the same neighbourhood in $X$ are equivalent). We call the algorithm of Lemma~\ref{lemma:BFSinDistCliquegraph} for each of these vertices as $v$, and also for each of the vertices of $X$ as $v$. We keep track of the largest distance found, and output it as the diameter after all calls. This results in an algorithm using $\O(2^k\ell \cdot k)$ passes and $\O((k+ \ell) \log n)$ bits of memory.
    
    The correctness follows from the correctness of Lemma~\ref{lemma:BFSinDistCliquegraph}, together with the fact that we consider all vertices as a start node, up to equivalence. This is because we can characterize each vertex by its adjacencies towards $X$ together with the clique it is contained in, which identifies the vertex up to equivalence on the closed neighbourhood.
\end{proof}

The performance of the algorithm in Theorem~\ref{thm:DiameterDistCliques} is distinctly worse than that of Theorem~\ref{thm:DiameterVC}, however, it does allow for more flexibility in the input in some specific cases. Note that the number of passes especially is exponential in $k$, but only linear in $\ell$, so a graph that is very close to a number of big cliques is well suited to apply this algorithm to.

For this setting, there also is a one pass but high memory approach.

\begin{restatable}{theorem}{Cliqueonepass}\label{thm:Clique1pass}
    Given a graph $G$ as an AL stream with deletion distance $k$ to $\ell$ cliques, with the given deletion set $X$, we can solve \textsc{Diameter [$k,\ell$]} using one pass and $\O(2^k\ell + (k + \ell) \log n)$ bits of memory.
\end{restatable}
\begin{proof}
    The approach here is similar to that of Theorem~\ref{thm:VC1pass}. In our pass, we save the deletion set $X$ and its internal edges. Next to this, every vertex not in $X$ can be characterized by its adjacencies towards $X$, together with what clique it is in. Therefore, each vertex can be characterized by $k + \log \ell$ bits, and for each option, we save whether we have 0, 1, or more of this vertex (this takes two bits). We need $\O(\ell \log n)$ bits to be able to identify which clique a vertex is in, as we can save a representative for each clique and identify a vertex by its adjacency to one of the representatives. We can find this information in our pass because it is an AL stream, and this takes $\O(2 \cdot 2^{k + \log \ell} + (k + \ell) \log n) = \O(2^k\ell + (k + \ell) \log n)$ bits of memory.
    
    To solve \textsc{Diameter}, we now only need to execute a procedure like that in Lemma~\ref{lemma:BFSinDistCliquegraph} for every pair of vertices, deciding on the distance between them. The theorem follows.
\end{proof}

We have now seen all the elements of Theorem~\ref{thm:DistlCliqueSummary}. We are not aware of any algorithms to compute the parameter distance to $\ell$ cliques.
\DistlCliqueSummary*
\begin{proof}
This follows immediately from Theorem~\ref{thm:DiameterDistCliques} and~\ref{thm:Clique1pass}.
\end{proof}


\section{Lower Bounds for Diameter}\label{sec:LowerBoundsDiameter}

We work with reductions from \textsc{Disj}$_n$, and we construct graphs where Alice controls some of the edges, and Bob controls some of the edges, depending on their respective input of the \textsc{Disj}$_n$ problem, and some parts of the graph are fixed. The aim is to create a gap in the diameter of the graph, that is, the answer to \textsc{Disj}$_n$ is YES if and only if the diameter is above or below a certain value. The lower bound then follows from Proposition~\ref{prop:DisjReduction}. Here $n$ may be the number of vertices in the graph construction, but may also be different (possibly forming a different lower bound). Our lower bounds hold for connected graphs. An overview of all hardness results for \textsc{Diameter} is given in Table~\ref{table:DisjDiameterAL}. 

\begin{table}[bt]
\centering
\begin{tabular}{|ll|l|l|l|}
\hline
 Parameter ($k$) / Graph class & Size & Bound & Theorem\\ \hline
 General and Bipartite Graphs & & \cellcolor{LightPink}$(\textsc{AL},n^2/p,p)$-hard & \ref{thm:QuadraticDiameter}, \ref{cor:QuadraticDiameterBipartite}\\ \hline
 \textsc{Vertex Cover} & $\geq 3$ & \cellcolor{LightPink}$(\textsc{VA}, n/p, p)$-hard & \ref{thm:SimpleVA}\\ \hline
 \textsc{Distance to $\ell$ cliques} & $k \geq 2, \ell \geq 1$ & \cellcolor{LightPink}$(\textsc{VA}, n/p, p)$-hard & \ref{thm:CliqueVA}\\\hline
 \textsc{FVS}, \textsc{FES} & $\geq 0$ & \cellcolor{LightPink}$(\textsc{AL},n/p,p)$-hard & \ref{thm:WindmillResults}\\ 
 & $\geq 0$ & \cellcolor{LightPink}$(\textsc{AL},n \log n,1)$-hard & \ref{cor:WindmillPermResults}\\\hline
 \textsc{Distance to matching} & $\geq 3$ & \cellcolor{LightPink}$(\textsc{AL},n/p,p)$-hard & \ref{cor:SimpleALResults}\\ \hline
 \textsc{Distance to path} & $\geq 2$ & \cellcolor{LightPink}$(\textsc{AL},n/p,p)$-hard & \ref{thm:DiamondResults}\\
 & $\geq 2$ & \cellcolor{LightPink}$(\textsc{AL},n \log n,1)$-hard & \ref{thm:DiamondPermResults}\\\hline
 \textsc{Distance to depth $\ell$ tree} & $k \geq 3, \ell \geq 2$ & \cellcolor{LightPink}$(\textsc{AL},n/p,p)$-hard & \ref{cor:SimpleALResults}\\ 
 & $k \geq 0, \ell \geq 5$ & \cellcolor{LightPink}$(\textsc{AL},n/p,p)$-hard & \ref{thm:WindmillResults}\\
 & $k \geq 0, \ell \geq 7$& \cellcolor{LightPink}$(\textsc{AL},n \log n,1)$-hard & \ref{cor:WindmillPermResults}\\\hline
 \textsc{Dist. to $\ell$ comps. of diam. $x$} & $k,x \geq 2$ & \cellcolor{LightPink}$(\textsc{AL},n/p,p)$-hard & \ref{cor:SimpleALResults}\\ \hline
 \textsc{Domination Number} & $\geq 3$ & \cellcolor{LightPink}$(\textsc{AL},n/p,p)$-hard & \ref{cor:SimpleALResults}\\ \hline
 \textsc{Maximum Degree} & $\geq 3$ & \cellcolor{LightPink}$(\textsc{AL},n/p,p)$-hard & \ref{thm:WindmillResults}\\
 & $\geq 3$ & \cellcolor{LightPink}$(\textsc{AL},n \log n,1)$-hard & \ref{cor:WindmillPermResults}\\\hline
 Split graphs & & \cellcolor{LightPink}$(\textsc{AL},n/p,p)$-hard & \ref{thm:Split}\\ \hline
\end{tabular}
\caption{An overview of the lower bounds for \textsc{Diameter}, with the parameter $(k)$ on the left. These results hold for connected graphs. $(\M, m, p)$-hard means that any algorithm using $p$ passes in model $\M$ (or weaker) requires $\Omega(m)$ bits of memory. \textsc{FVS} stands for Feedback Vertex Set number, \textsc{FEN} for Feedback Edge Set number.}
\label{table:DisjDiameterAL}
\end{table}

We start by proving simple lower bounds for the VA model when our problem is parameterized by the vertex cover number, and when our problem is parameterized by the distance to $\ell$ cliques. This shows that we actually need the AL model to achieve the upper bounds in Section~\ref{sec:upperboundsdiameter}. The constructions are illustrated in Figure~\ref{fig:SimpleVA} and Figure~\ref{fig:CliqueVA}. Generally, $a$-vertices ($b$-vertices) and their incident edges are controlled by Alice (Bob).

\begin{figure}
\centering
\begin{minipage}{.48\textwidth}
  \centering
  \includegraphics[width=.6\linewidth]{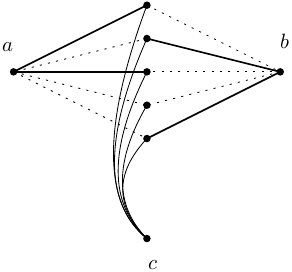}
  \captionof{figure}{VA lower bound for diameter with vertex cover size 3, called `Simple VA'. The vertices in the middle are indexed $1,\ldots,n$. An edge incident to $a$ ($b$) is present when the entry of Alice (Bob) at the corresponding index is 1. The vertex $c$ ensures the graph is connected.}
  \label{fig:SimpleVA}
\end{minipage}\quad
\begin{minipage}{.48\textwidth}
  \centering
  \includegraphics[width=.8\linewidth]{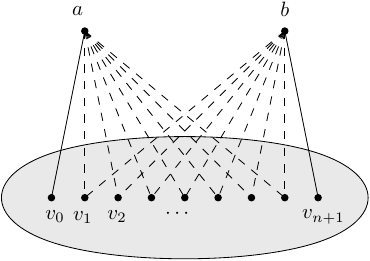}
  \captionof{figure}{VA lower bound for diameter with distance 2 to $1$ clique, called `Clique VA'. A dashed edge is present when the entry at the corresponding index is 1. The vertices inside the grey area form a clique. Hence, deletion distance to a clique is~$2$ (remove $a$ and $b$).}
  \label{fig:CliqueVA}
\end{minipage}
\end{figure}

\begin{restatable}{theorem}{SimpleVA}\label{thm:SimpleVA}
    Any streaming algorithm for \textsc{Diameter} on graphs of vertex cover number at least 3 in the VA model that uses $p$ passes over the stream requires $\Omega(n/p)$ bits of memory.
\end{restatable}
\begin{proof}
    Let $x,y$ be the input to $\textsc{Disj}_n$ of Alice and Bob, respectively. Assume we have a streaming algorithm for \textsc{Diameter} in the VA model. We construct a graph as illustrated in Figure~\ref{fig:SimpleVA}. First Alice reveals to the stream $n$ vertices $v_1, \ldots, v_n$ with no edges, then she reveals the vertex $c$ which is connected to all the $n$ vertices. Now we associate the $n$ vertices with the indices of $x$ and $y$, associate vertex $v_i$ with index $i$. Alice reveals a vertex $a$, and for each index $i$ she reveals an edge between $a$ and $v_i$ when the entry at the index $i$ in $x$ is a 1. After this, she passes the memory state of the algorithm to Bob. Bob now reveals a vertex $b$ and similar to Alice, reveals an edge between $b$ and $v_i$ when the entry at index $i$ in $y$ is a 1. This completes the construction of the graph, and thus the stream. Let it be clear that this is a VA stream that Alice and Bob can construct without knowing input of the other. The graph is always connected because if either Alice or Bob has an all-zeroes input, the problem of $\textsc{Disj}_n$ is trivially solvable (so $\textsc{Disj}_n$ is equally hard ignoring this case)\footnote{In fact, there are stronger assumptions that one can make while the complexity of $\textsc{Disj}_n$ remains the same, see \cite{ChakrabartiDisjAssumptions} or \cite{AgarwalSpatialScan}.}.
    
    We now claim that the diameter of this graph is at most 3 when the answer to $\textsc{Disj}_n$ is NO, and otherwise the diameter is 4.
    
    Let us assume that the answer to $\textsc{Disj}_n$ is NO, that is, there is an index $i$ such that $x_i = y_i = 1$. Then clearly, the distances between $a, b$ and $c$ are all 2, by viewing the paths using $v_i$ as an intermediate vertex. Hence the diameter is at most 3, which can be formed by some path from e.g. $a$ to $v_i$ to $c$ to some $v_j$ that is non-adjacent to $a$.
    
    Now assume the answer to $\textsc{Disj}_n$ is YES, that is, there is no index $i$ such that $x_i = y_i = 1$. Now consider the distance between $a$ and $b$. To get from $a$ to $b$, we need to go from $a$ to some $v_i$ (which is non-adjacent to $b$ by the assumption), then go to $c$ and to another $v_j$ which is adjacent to $b$ (but not $a$), to go to $b$. This path has length 4, and must exist by the non-all-zeroes input assumption, and forms a shortest path from $a$ to $b$ in this graph. So the diameter of the graph is 4.
    
    In conclusion, we constructed a connected graph of size $n+3$ with a vertex cover number of 3 (taking $\{a,b,c\}$ suffices) that can be given as a VA stream to a streaming algorithm for \textsc{Diameter} to solve the $\textsc{Disj}_n$ problem. The theorem follows from Proposition~\ref{prop:DisjReduction}.
\end{proof}

\begin{restatable}{theorem}{CliqueVA}\label{thm:CliqueVA}
    Any streaming algorithm for \textsc{Diameter} on graphs of distance 2 to $\ell = 1$ clique in the VA model that uses $p$ passes over the stream requires $\Omega(n/p)$ bits of memory.
\end{restatable}
\begin{proof}
    Let $x,y$ be the input to $\textsc{Disj}_n$ of Alice and Bob, respectively. Assume we have a streaming algorithm for \textsc{Diameter} in the VA model. We construct a graph as illustrated in Figure~\ref{fig:CliqueVA}. Start with a clique on $n+2$ vertices, $v_0, \ldots, v_{n+1}$. Let $a,b$ be two vertices not in the clique, and add the edges $(a,v_0)$ and $(b,v_{n+1})$. Then, for any $i$, Alice adds the edge $(a,v_i)$ when $x_i=1$ and Bob adds the edge $(b,v_i)$ when $y_i = 1$. This completes the construction. Alice and Bob construct the VA stream as follows. First Alice reveals $v_0, \ldots, v_{n+1}$ and then reveals $a$ (for which she knows what edges should be present). Then Alice passes the memory of the algorithm to Bob, who reveals $b$, which completes the stream. Notice that the graph is always connected by the fixed edges $(a,v_0)$ and $(b,v_{n+1})$.
    
    We now claim that the diameter of this graph is 2 when the answer to $\textsc{Disj}_n$ is NO, and otherwise the diameter is at least 3.
    
    Let us assume that the answer to $\textsc{Disj}_n$ is NO, that is, there is an index $i$ such that $x_i = y_i = 1$. Notice how the distance between $a$ and $b$ is now 2, because both are connected to $v_i$. The distance between any other pair of vertices is also at most 2, because all vertices except $a$ and $b$ form a clique.
    
    Now assume the answer to $\textsc{Disj}_n$ is YES, that is, there is no index $i$ such that $x_i = y_i = 1$. The shortest path between $a$ and $b$ in this instance must use some edge in the clique, as these vertices do not have a common neighbour. Hence, the distance between $a$ and $b$ is at least 3.
    
    In conclusion, we constructed a connected graph of size $n+4$ with a distance 2 to $\ell = 1$ clique (taking $\{a,b\}$ suffices) that can be given as a VA stream to a streaming algorithm for \textsc{Diameter} to solve the $\textsc{Disj}_n$ problem. The theorem follows from Proposition~\ref{prop:DisjReduction}.
\end{proof}

The lower bounds in Figure~\ref{fig:SimpleVA} and Figure~\ref{fig:CliqueVA} do not work for the AL model because there are vertices that may be adjacent to both $a$ and $b$, so neither Alice nor Bob can produce the adjacency list of such a vertex alone. For the `Simple VA' construction, we can `fix' this by extending these vertices to edges but this is destructive to the small vertex cover number of the construction, see Figure~\ref{fig:SimpleAL}. It should be clear that AL reductions require care: the set of edges incident to a vertex has to be fully determined when Alice or Bob wants to reveal it.

\begin{figure}
  \centering
  \includegraphics[width=.6\linewidth]{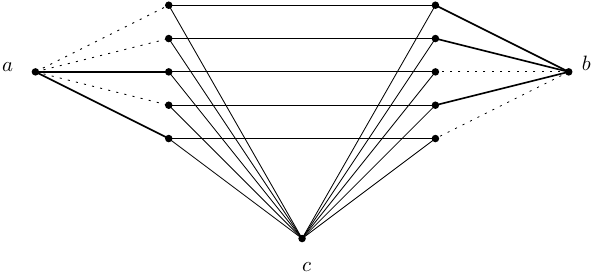}
  \captionof{figure}{AL lower bound for diameter, called `Simple AL'. The edges in the middle are indexed $1,\ldots,n$. An edge incident on $a$ ($b$) is present when the entry of Alice (Bob) at the corresponding index is 1. The vertex $c$ ensures the graph is connected.}
  \label{fig:SimpleAL}
\end{figure}

\begin{restatable}{theorem}{SimpleAL}\label{thm:SimpleAL}
    Any streaming algorithm for \textsc{Diameter} that works on the `Simple AL' construction in the AL model using $p$ passes over the stream requires $\Omega(n/p)$ bits of memory.
\end{restatable}
\begin{proof}
    Let $x,y$ be the input to $\textsc{Disj}_n$ of Alice and Bob, respectively. We construct a graph as illustrated in Figure~\ref{fig:SimpleAL}. The graph consist of $2n + 3$ vertices. This is a matching $M$ on $2n$ vertices. We make a vertex $c$ adjacent to all vertices of $M$. Next to this, we have vertices $a$ and $b$ of which the adjacencies towards one end of $M$ are dependent on the input of Alice and Bob, respectively. An edge between $a$ and $i$th edge on the `left' side of $M$ is present when Alice has a 1 on index $i$ in $x$. An edge between $b$ and $i$th edge the `right' side of $M$ is present when Bob has a 1 on index $i$ in $y$. Assuming we have an algorithm that can work on this graph, to construct an AL stream for the algorithm, Alice can first reveal $a, c$, and all vertices on the left side of $M$, then pass the memory of the algorithm to Bob, who reveals $b$ and all the vertices on the right side of $M$. This completes one pass of the stream. The graph is always connected because if either Alice or Bob has an all-zeroes input, the problem of $\textsc{Disj}_n$ is trivially solvable (so $\textsc{Disj}_n$ is equally hard ignoring this case).
    
    We now claim that the diameter of the `Simple AL' graph is at most 3 when the answer to $\textsc{Disj}_n$ is NO, and otherwise the diameter is at least 4.
    
    Let us assume that the answer to $\textsc{Disj}_n$ is NO, that is, there is an index $i$ such that $x_i = y_i = 1$. Notice that the distance from $c$ to any vertex is at most 2. The distance from $a$ to $b$ is 3 by taking the $i$th matching edge. The distances from $a$ to any other vertex is at most 3 by going through $c$, and the same holds for $b$. So the diameter is at most 3.
    
    Now assume the answer to $\textsc{Disj}_n$ is YES, that is, there is no index $i$ such that $x_i = y_i = 1$. Consider the distance between $a$ and $b$. As there is no index where both have a 1, the shortest path from $a$ to $b$ must use $c$ as an intermediate vertex. But then this path has length at least 4. So the diameter is at least 4.
    
    We conclude that any algorithm that can solve the \textsc{Diameter} problem on the graph construction `Simple AL' in the AL model in $p$ passes, must use $\Omega(n/p)$ bits of memory by Proposition~\ref{prop:DisjReduction}.
\end{proof}

The following follows from Theorem~\ref{thm:SimpleAL} by observing some properties of the `Simple AL' construction in Figure~\ref{fig:SimpleAL}.

\begin{restatable}{corollary}{SimpleALResults}\label{cor:SimpleALResults}
    Any streaming algorithm for \textsc{Diameter} in the AL model that uses $p$ passes over the stream must use $\Omega(n/p)$ bits of memory, even on graphs for which the algorithm is given a 
    \begin{enumerate}[noitemsep]
        \item Deletion Set to Matching of size at least 3,
        \item Deletion Set to $\ell$ components of diameter $x$ of size at least 2, $x\geq 2$,
        \item Dominating Set of size at least 3,
        \item Deletion Set to a depth $\ell$ tree of size at least 3, $\ell \geq 2$.
    \end{enumerate}
\end{restatable}
\begin{proof}
    The corollary follows from Theorem~\ref{thm:SimpleAL}, together with observing that the construction of `Simple AL' is
    \begin{enumerate}[noitemsep]
        \item a matching when removing $\{a, b, c\}$,
        \item one component of diameter 2 when removing $\{a,b\}$,
        \item dominated by the set of vertices $\{a,b,c\}$,
        \item a tree of depth 2 when we add a vertex $d$ adjacent to the left part of $M$\footnote{Notice that adding $d$ does not change any of the distances.}, and remove $\{a, b, c\}$.
    \end{enumerate}
\end{proof}

Next, we construct a lower bound for a special case, when the input graph is a tree, see Figure~\ref{fig:windmill}.

\begin{figure}[!b]
    \centering
    \includegraphics[width=.7\textwidth]{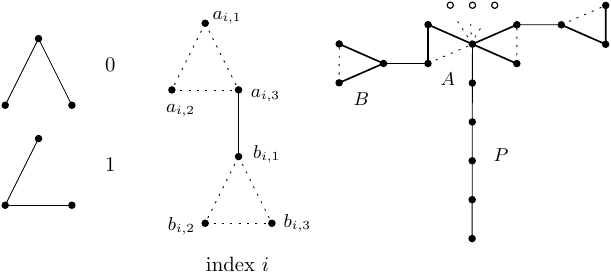}
    \caption{AL lower bound for diameter consisting of a tree, called `Windmill'. The difference in an entry 1 or 0 is shown on the left. The gadget for index $i$ combines a 0/1-gadget for Alice and a 0/1-gadget for Bob. It makes two 1 entries at this index a path of length 5, and a tree structure of depth at most 4 otherwise. These $n$ gadgets are then identified at $a_{i,1}$ and a tail is added.}
    \label{fig:windmill}
\end{figure}

\begin{restatable}{theorem}{Windmill}\label{thm:Windmill}
    Any streaming algorithm for \textsc{Diameter} that works on the `Windmill' construction in the AL model using $p$ passes over the stream requires $\Omega(n/p)$ bits of memory.
\end{restatable}
\begin{proof}
    Let $x,y$ be the input to $\textsc{Disj}_n$ of Alice and Bob, respectively. We construct a graph as illustrated in Figure~\ref{fig:windmill}. The graph consist of $5n + 6$ vertices. We start with a path $P$ on $6$ vertices, call the vertex on one of the ends the \emph{center}. To this center, we will `glue' $n$ gadgets, which may vary depending on the input of Alice and Bob, so associate an index $i$ with each gadget. Each gadget adds 5 vertices to the graph. Let us describe one such gadget. Consider two triplets of vertices $a_{i,1},a_{i,2},a_{i,3}$ (Alice) and $b_{i,1},b_{i,2},b_{i,3}$ (Bob), and connect $a_{i,3}$ to $b_{i,1}$ with an edge. For Alice, if the entry at index $i$ is a 0, she inserts the edges $(a_{i,1},a_{i,2})$ and $(a_{i,1},a_{i,3})$, and if it is a 1, she inserts $(a_{i,1},a_{i,2})$ and $(a_{i,2},a_{i,3})$. Bob does the same for his triplet. We `glue' this gadget into the graph by identifying $a_{i,1}$ to be the same vertex as the center vertex in the graph. Assuming we have an algorithm that works on this graph, to construct an AL stream containing this graph, Alice first reveals the path $P$, and all her own vertices $a_{i,2}, a_{i,3}$ for all $i$, with all the incident edges. Then she passes the memory of the algorithm to Bob who reveals the vertices $b_{i,1}, b_{i,2}, b_{i,3}$ for all $i$, with all the incident edges. This completes one pass of the stream. Notice that the graph is connected.
    
    We now claim that the diameter of the `Windmill' graph is at least 10 when the answer to $\textsc{Disj}_n$ is NO, and otherwise the diameter is at most 9.
    
    Let us assume that the answer to $\textsc{Disj}_n$ is NO, that is, there is an index $i$ such that $x_i = y_i = 1$. Then, the distance from the end of the path $P$ to $b_{i,3}$ is exactly 10, and this is the only simple path between these vertices, so it is the shortest path. Hence, the diameter is at least 10.
    
    Now assume the answer to $\textsc{Disj}_n$ is YES, that is, there is no index $i$ such that $x_i = y_i = 1$. Then, the distance from the center vertex to any other vertex $a_{i,j}$ or $b_{i,j}$ is at most 4, as at least one of the triplets for each index forms a tree-like shape and not a path. Therefore, the diameter is at most 9, formed by the shortest path from the end of the path $P$ to some $b_{i,3}$.
    
    We conclude that any algorithm that can solve the \textsc{Diameter} problem on the graph construction `Windmill' in the AL model in $p$ passes, must use $\Omega(n/p)$ bits of memory by Proposition~\ref{prop:DisjReduction}.
\end{proof}

The following follows from Theorem~\ref{thm:Windmill} by observing properties of the `Windmill' construction.

\begin{restatable}{corollary}{WindmillResults}\label{thm:WindmillResults}
    Any streaming algorithm for \textsc{Diameter} in the AL model that uses $p$ passes over the stream must use $\Omega(n/p)$ bits of memory, even on graphs for which the algorithm is given
    \begin{enumerate}[noitemsep]
        \item that the input is a bounded depth tree,
        \item that the Maximum Degree is a constant of at least 3.
    \end{enumerate}
\end{restatable}
\begin{proof}
    The corollary follows from Theorem~\ref{thm:Windmill}, together with observing that `Windmill' is
    \begin{enumerate}[noitemsep]
        \item a bounded depth tree,
        \item a lower bound that still works when we convert the center vertex into a binary tree of depth $c = \O(\log n)$ and extend the path to size $5 + c$ accordingly (this makes the diameter distinction to be $9 + 2c$ or $10 + 2c$).
    \end{enumerate}
\end{proof}

Next we look at the case when the input graph is close to a path.

\begin{figure}[!htp]
    \centering
    \includegraphics[width=\textwidth]{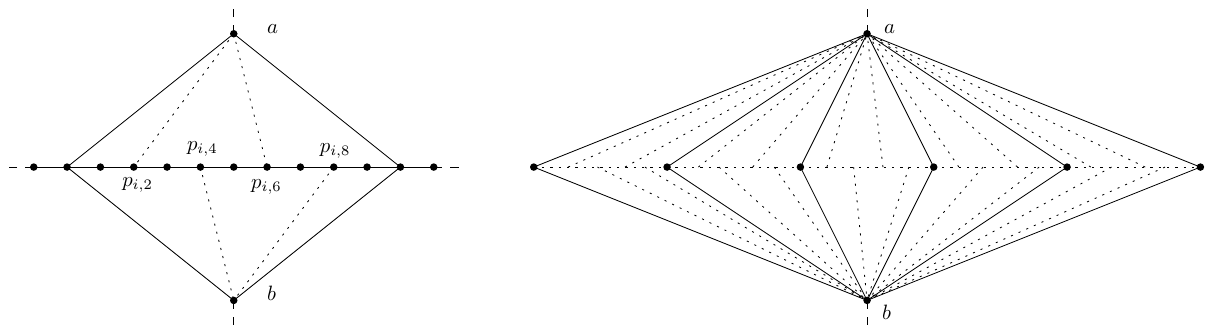}
    \caption{AL lower bound for diameter consisting of a path and 2 vertices, called `Diamond'. Note that $a$ is connected to $b$ with an edge (indicated with a dashed line here). On the left the gadget for a single index $i$ is shown, where the dotted edges are present when the entry at index $i$ is 0 (for Alice incident on $a$, for Bob incident on $b$). On the right, the construction is sketched in full.}
    \label{fig:diamond}
\end{figure}

\begin{restatable}{theorem}{Diamond}\label{thm:Diamond}
    Any streaming algorithm for \textsc{Diameter} that works on the `Diamond' construction in the AL model using $p$ passes over the stream requires $\Omega(n/p)$ bits of memory.
\end{restatable}
\begin{proof}
    Let $x,y$ be the input to $\textsc{Disj}_n$ of Alice and Bob, respectively. We construct a graph as illustrated in Figure~\ref{fig:diamond}. We construct a graph on $10n+3$ vertices. We start with two vertices $a$ and $b$, connected by an edge. Create $n + 1$ vertices connected to $a$ and $b$ with an edge and label them $c_0, \ldots c_n$. We create a gadget for index $i$ between $c_{i-1}$ and $c_i$ for every $1\leq i \leq n$. For index $i$, we insert a path $P_i = p_{i,1}, \ldots, p_{i,9}$ on 9 vertices with $p_{i,1}$ connected to $c_{i-1}$ and $p_{i,9}$ to $c_i$ with an edge. The edges $(a,p_{i,2})$ and $(a,p_{i,6})$ are present if and only if $x_i = 0$ for Alice, and the edges $(b, p_{i,4})$ and $(b,p_{i,8})$ are present if and only if $y_i = 0$ for Bob. Assuming we have an algorithm that works on this construction, to construct an AL stream containing this graph, Alice first reveals all vertices except $b$ and $p_{i,4}$ and $p_{i,8}$ for all $i$ (the edges are fixed, or the input of Alice decides the edges, for these vertices), then passes the memory state of the algorithm to Bob who reveals exactly $b$ and $p_{i,4}$ and $p_{i,8}$ for all $i$. This completes one pass of the stream. Notice that the graph is connected.
    
    We now claim that the diameter of the `Diamond' graph is at least 8 when the answer to $\textsc{Disj}_n$ is NO, and otherwise the diameter is at most 7.
    
    Let us assume that the answer to $\textsc{Disj}_n$ is NO, that is, there is an index $i$ such that $x_i = y_i = 1$. Then, consider $p_{i,5}$. The distance from this vertex to $a$ or $b$ is exactly 6. So, the distance from $p_{i,5}$ to some other $p_{j,5}$ for $j\neq i$ is at least 8.
    
    Now assume the answer to $\textsc{Disj}_n$ is YES, that is, there is no index $i$ such that $x_i = y_i = 1$. Then for any $1\leq i \leq n$ and $1\leq j\leq 9$ the distance from $p_{i,j}$ to one of either $a$ or $b$ is at most 3. So every vertex in the graph has distance at most 3 to either $a$ or $b$. But then, as $a$ and $b$ are connected with an edge, the diameter must be at most 7.
    
    We conclude that any algorithm that can solve the \textsc{Diameter} problem on the graph construction `Diamond' in the AL model in $p$ passes, must use $\Omega(n/p)$ bits of memory by Proposition~\ref{prop:DisjReduction}.
\end{proof}

The following follows from Theorem~\ref{thm:Diamond} by observing properties of the `Diamond' construction.

\begin{restatable}{corollary}{DiamondResults}\label{thm:DiamondResults}
    Any streaming algorithm for \textsc{Diameter} in the AL model that uses $p$ passes over the stream must use $\Omega(n/p)$ bits of memory, even on graphs for which the algorithm is given a Deletion Set to a path of size at least 2.
\end{restatable}
\begin{proof}
    The corollary follows from Theorem~\ref{thm:Diamond}, together with observing that `Diamond' is a path when we remove $\{a,b\}$.
\end{proof}

\begin{figure}[!htp]
    \centering
    \includegraphics[width=.6\textwidth]{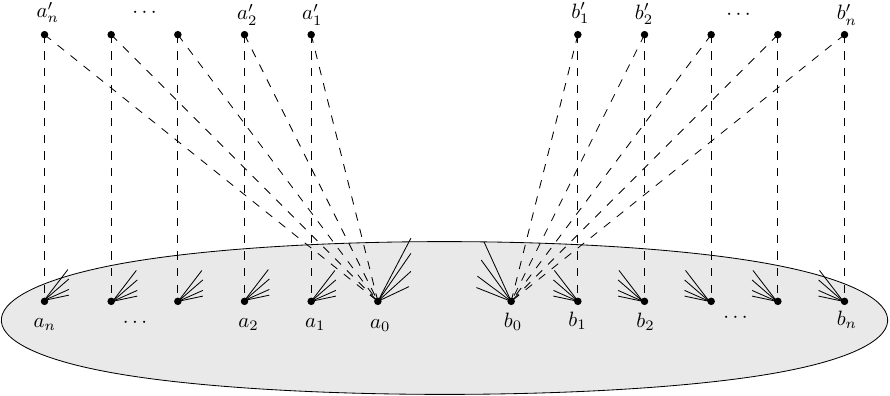}
    \caption{AL lower bound for diameter on split graphs, called `Split'. Depending on the input, some $a'_i$ either has an edge to $a_0$ or $a_i$ when $x_i = 0$ or 1. The same holds for $b'_i$ with $y_i$. The grey area forms a clique, and each $a_i$ is connected to all $b'_j$ where $i\neq j$, and the same holds for $b_i$ and $a'_j$.}
    \label{fig:split}
\end{figure}

Next we show that the \textsc{Diameter} problem in the AL model is hard on split graphs.

\begin{restatable}{theorem}{Split}\label{thm:Split}
    Any streaming algorithm for \textsc{Diameter} that works on split graphs in the AL model using $p$ passes over the stream requires $\Omega(n/p)$ bits of memory.
\end{restatable}
\begin{proof}
    Let $x,y$ be the input to $\textsc{Disj}_n$ of Alice and Bob, respectively. We construct a graph as illustrated in Figure~\ref{fig:split}. We construct a graph on $4n+2$ vertices. The split graph we construct has as a clique on $2n+2$ vertices, let this be the vertices $a_0, \ldots, a_n$, $b_0, \ldots, b_n$. The independent set consists of the vertices $a_1',\ldots, a_n'$, $b_1',\ldots, b_n'$. The following edges are present regardless of the input: we connect $a_0$ to all $b_i'$, $1\leq i \leq n$, and similarly $b_0$ to all $a_i'$, $1\leq i \leq n$. We also connect each $a_i$ to all $b_j'$ where $j\neq i$ for $1 \leq i \leq n$. Similarly, we connect each $b_i$ to all $a_j'$ where $j\neq i$ for $1 \leq i \leq n$. The other edges are input dependent. For each index $i$, the edge $(a_i, a_i')$ is inserted when $x_i = 1$ and otherwise the edge $(a_0, a_i')$ is inserted. Similarly, for each $i$, the edge $(b_i, b_i')$ is inserted when $y_i = 1$ and otherwise the edge $(b_0, b_i')$ is inserted. This completes the construction. Note that it is a split graph as the vertices $a_0, \ldots, a_n$, $b_0, \ldots, b_n$ form a clique and $a_1',\ldots, a_n'$, $b_1',\ldots, b_n'$ form an independent set. Assuming we have an algorithm that works on split graphs, to construct an AL stream containing this graph, Alice reveals all the $a$-vertices. Then she passes the memory of the algorithm to Bob, who reveals all the $b$-vertices. This completes one pass of the stream. Notice that Alice and Bob do not require information on the input of the other, as only input-independent edges connect $a$-vertices to $b$-vertices.
    
    We claim that the diameter of this graph is at most 2 if the answer to $\textsc{Disj}_n$ is YES and otherwise the diameter is at least 3.
    
    Let us assume the answer to $\textsc{Disj}_n$ is NO, that is, there is an index $i$ such that $x_i = y_i = 1$. Consider the distance between $a_i'$ and $b_i'$ in this instance. Notice that because $x_i = y_i = 1$, there is no vertex in the clique connected to both vertices. As both vertices are connected to some vertex in the clique, the distance between them must be 3.
    
    Now assume the answer to $\textsc{Disj}_n$ is YES, that is, there is no index $i$ such that $x_i = y_i = 1$. We show that the diameter is at most 2. The distance between vertices in the clique is at most 1. The distance from any $a_i'$ or $b_i'$ to some vertex in the clique is at most 2, because each $a_i'$ or $b_i'$ is always connected to at least one vertex in the clique. The distance from any $a_i'$ to another $a_j'$ is 2 because of $b_0$, similarly, the distance from any $b_i'$ to another $b_j'$ is 2 because of $a_0$. Let $1 \leq i,j \leq n$ be two (possibly the same) indices, and consider the distance between $a_i'$ and $b_j'$. If either $x_i = 0$ or $y_j = 0$ then the distance is 2 because of $a_0$ or $b_0$. Otherwise, both have a 1 at the corresponding index. But then we know that $i \neq j$, and so $a_i', a_i, b_j'$ is a path in the graph of length 2. Hence, the diameter of the graph is at most 2.
    
    We conclude that any algorithm that can solve the \textsc{Diameter} problem on split graphs in the AL model in $p$ passes, must use $\Omega(n/p)$ bits of memory by Proposition~\ref{prop:DisjReduction}.
\end{proof}

We can now prove Theorem~\ref{thm:HfreeOverview} and Theorem~\ref{thm:DistHfreeOverview}. Intuitively, if $H$ contains a cycle or a vertex of degree~$3$, a modification of `Windmill' is $H$-free; if $H$ is a linear forest, a modification of `Split' is (almost) $H$-free.

\HfreeOverview*
\begin{proof}
If $H$ contains a cycle as a subgraph, then the result follows from Theorem \ref{thm:Windmill}. Hence, we may assume that $H$ does not contain a cycle as a subgraph, and thus is a forest.

If $H$ contains a tree with a vertex of degree at least three, then the result follows from a slight modification of Theorem \ref{thm:Windmill}. Start from the version of the construction where each vertex has degree at most~$3$ (per Corollary \ref{thm:WindmillResults}) and let $c$ denote the center vertex. Note that the diametral pair in the construction is the other end $v\not=c$ of the path $P$ (recall that $c$ is one of its ends) and another leaf of the tree. Hence, we can add edges (which shorten distances) as long as this property is preserved. Consider the tree to be rooted at $v$. Make the two children of $c$ (which are not on $P$ by the choice of the root) adjacent, and recurse down the tree, consistently making children adjacent if there are two. The resulting graph has no $K_{1,3}$ as an induced subgraph, and thus is $H$-free. Hence, we may assume that $H$ also does not contain a vertex of degree at least three, and thus is a linear forest.

We now reduce the open cases to just $H=4P_1$ and $H=P_4+P_1$ and later show hardness for those cases. If $H$ contains a $2P_2$ as an induced subgraph, then the result follows from Theorem \ref{thm:Split}, as split graphs are $2P_2$-free. Hence, $H$ does not contain $2P_2$ as an induced subgraph. In particular, we may assume that all paths in $H$ are of length at most~$3$. If $H$ is the union of a $P_4$ and either at least two other paths or another path of length at least~$2$, then it contains a $4P_1$. In the other cases when $H$ is the union of a $P_4$ and other paths, it is $P_4+P_2$ (which contains a $2P_2$, and thus was already excluded) or $P_4+P_1$.
If $H$ is the union of a $P_3$ and at least two other paths, it contains a $4P_1$. If $H$ is the union of a $P_3$ and another path of length at least~$1$, then it contains a $2P_2$ and thus was already excluded. The case when $H$ is $P_3+P_1$ has been excluded by assumption. Hence, we may assume that $H$ contains only paths of length at most~$1$.
If $H$ contains a $P_2$, then it cannot contain another $P_2$, as $2P_2$ would be an induced subgraph, nor can it be $P_2+P_1$ which is an induced subgraph of $P_4$, nor can it be $P_2+sP_1$ for $s \geq 2$ which is $P_2+2P_1$ (excluded by assumption) or contains a $4P_1$. If $H$ does not contain a $P_2$, then it is $sP_1$ for some $s$. However, $P_1$ and $2P_1$ are induced subgraphs of $P_4$, $3P_1$ is excluded by assumption, and $sP_1$ for $s \geq 4$ contains $4P_1$ as an induced subgraph. Hence, the open cases have been successfully reduced.

To tackle the remaining cases, $H=4P_1$ or $H=P_4+P_1$,  we modify the construction of Theorem \ref{thm:Split}. Let $(C,I)$ be the split partition implied by the construction. In that construction, it can be readily seen that the vertices $a'_1,\ldots,a'_n$ can be turned into a clique $A'$ and the vertices $b'_1,\ldots,b'_n$ can be turned into a clique $B'$ without affecting the correctness of the reduction. Observe that the resulting graph is $4P_1$-free, as it is a union of three cliques. To see that it is also $P_4+P_1$-free, note that in any induced subgraph isomorphic to $P_4+P_1$, the $P_4$ must contain two consecutive vertices in $A'$, say $a'_i$ and $a'_j$ for some $i \not= j$, and two consecutive vertices $c,c'$ in $C$ (the case when it contains two consecutive vertices in $B'$ and $C$ is symmetric). Note that $c,c' \not\in \{b_0,b_1,\ldots,b_n\}$ or $c'$ (the end of the $P_4$) would be adjacent to $a'_i$ or $a'_j$. Moreover, $c,c' \not\in \{a_0,a_1,\ldots,a_n\}$, as they would jointly cover $B'$, leaving no room for the $P_1$. The theorem follows.
\end{proof}

\DistHfreeOverview*
\begin{proof}
We start with the case when $G-X$ has to be a connected graph. By Theorem \ref{thm:HfreeOverview}, only the cases when $H \subseteq_i P_4$ and $H= 3P_1$, $H =P_3+P_1$, $H=P_2+2P_1$ still need to be proven.

If $H$ is a $P_4$, then the result follows from Theorem \ref{thm:SimpleAL}, because in that construction for some $X$ of size~$2$, $G-X$ is a union of triangles where the triangles have a single common vertex $c$. Hence, we may assume $H$ has only paths of length at most $2$. If $H$ is $P_3+P_1$, then the result follows again from Theorem \ref{thm:SimpleAL}, because in that construction the vertex $c$ is dominating, yet must be in any induced $P_3$. We also note that by assumption, $H \not= P_3$. Hence, we may assume $H$ only has paths of length at most $1$. These cases are all resolved by assumption.

In case $G-X$ does not have to be connected, the only relevant case is when $H = P_3$. In that case, the result follows from Theorem \ref{thm:SimpleAL}, because in that construction for some $X$ of size~$3$, $G-X$ is a matching.
\end{proof}

We can also prove a quadratic bound for general graphs; see Figure~\ref{fig:quadratic} for the construction.

\begin{figure}
    \centering
    \includegraphics[width=.6\textwidth]{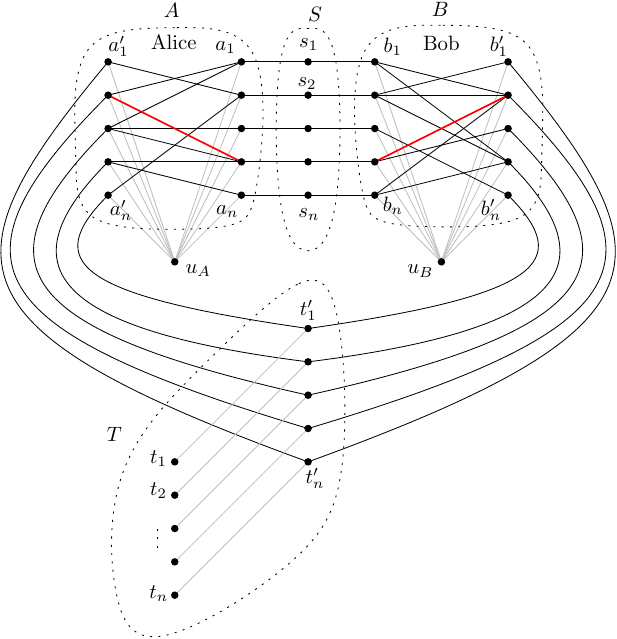}
    \caption{AL lower bound for diameter where Alice and Bob have $n^2$ bits but the graph has $\O(n)$ vertices. The bits are seen as an adjacency matrix in the bipartite graphs $A$ and $B$, identically: the red edge $a_i'$ to $a_j$ in $A$ is the same index as the red edge $b_j$ to $b_i'$ in $B$. Edges are present when the entry is a 0. Then, each $s_i, t_j$ pair can discern whether or not at least one of the edges $a_i$ to $a'_j$ or $b_i$ to $b'_j$ is present, hence deciding whether or not both Alice and Bob have a 1 at that entry.}
    \label{fig:quadratic}
\end{figure}

\begin{restatable}{theorem}{QuadraticDiameter}\label{thm:QuadraticDiameter}
    Any streaming algorithm for \textsc{Diameter} on general (dense) graphs in the AL model using $p$ passes over the stream requires $\Omega(n^2/p)$ bits of memory.
\end{restatable}
\begin{proof}
    Let $N = n^2$. Let $x,y$ be the input to $\textsc{Disj}_N$ of Alice and Bob, respectively. We construct a graph as illustrated in Figure~\ref{fig:quadratic}. We construct a graph on $7n + 2$ vertices, but $\O(n^2)$ edges. The input of Alice and Bob will control the edges in a bipartite graph each. Let $a_1, \ldots a_n$ and $a_1', \ldots, a_n'$ be the bipartite graph $A$ for Alice. Alice now views her input $x$ as an adjacency matrix for the $n^2$ potential edges in $A$, but inverse, so an edge is present if and only if the corresponding entry is a 0. We also add a (universal) vertex $u_A$ which we connect to all vertices in $A$. For Bob, do the same with vertices $b_1,\ldots,b_n$ and $b_1',\ldots,b_n'$ forming a bipartite graph $B$. Bob also views his input $y$ as an adjacency matrix (in exactly the same order as Alice!) for the $n^2$ potential edges in $B$, but inverse, so an edge is present if and only if the corresponding entry is a 0. We also add a vertex $u_B$ which we connect to all vertices in $B$. To complete the construction, we create a set $S$ of $n$ vertices $s_1,\ldots,s_n$ and we connect $s_i$ to $a_i$ and $b_i$ for each $1\leq i\leq n$. Then, we also create a set $T$ of $2n$ vertices $t_1,\ldots,t_n$ and $t_1',\ldots,t_n'$, where we connect $t_i'$ with $t_i$ and $a_i'$ and $b_i'$. This completes the construction. Given an algorithm which works on such a graph, Alice and Bob can construct an AL stream by having Alice first reveal all vertices in $S,A, T$ and $u_A$ with their incident edges, then passing the memory state to Bob who reveals all vertices in $B$ and $u_B$ with their incident edges. This completes one pass of the stream. Notice that the graph is connected.
    
    We now claim that the diameter of this construction is at least 5 when the answer to $\textsc{Disj}_N$ is NO, and otherwise the diameter is at most 4.
    
    Let us assume that the answer to $\textsc{Disj}_n$ is NO, that is, there is an index $i'$ such that $x_{i'} = y_{i'} = 1$. Let $i, j$ be the pair of indices such that the edges $(a_i,a_j')$ and $(b_i,b_j')$ are decided by $x_{i'}$ and $y_{i'}$ respectively. In this case, both these edges are not present in the graph. Hence, the shortest path from $s_i$ to $t_j$ must use either $u_A$ or $u_B$, or use at least 3 edges in $A$ or $B$ (because $A$ and $B$ are bipartite graphs). Hence, the distance from $s_i$ to $t_j$ must be at least 5.
    
    Now assume the answer to $\textsc{Disj}_n$ is YES, that is, there is no index $i'$ such that $x_{i'} = y_{i'} = 1$. Then, for every $1\leq i,j\leq n$ pair there exists a path from $s_i$ to $t_j$ of length 4, because either or both of the edges $(a_i,a_j')$, $(b_i,b_j')$ are present in the graph. One can check that all other distances are at most 4 as well\footnote{If the reader is not convinced, notice that we could always extend the tails $t_i,t_i'$ to consist of a longer path, making paths other that those that originate from $T$ negligible.}.
    
    We conclude that any algorithm that can solve the \textsc{Diameter} problem on general (dense) graphs in the AL model in $p$ passes, must use $\Omega(N/p) = \Omega(n^2/p)$ bits of memory by Proposition~\ref{prop:DisjReduction}.
\end{proof}

Splitting up $u_A$ and $u_B$ into two vertices each, and making the tails from $t'_i$ to $t_i$ at least three edges longer for each $i$ makes the lower bound work for bipartite graphs.

\begin{restatable}{corollary}{QuadraticDiameterBipartite}\label{cor:QuadraticDiameterBipartite}
    Any streaming algorithm for \textsc{Diameter} on bipartite graphs in the AL model using $p$ passes over the stream requires $\Omega(n^2/p)$ bits of memory.
\end{restatable}
\begin{proof}
    The proof follows from adjusting the construction in Theorem~\ref{thm:QuadraticDiameter}. If we split up $u_A$ and $u_B$ into two vertices $u_A, u_A'$ and $u_B,u_B'$ each and have each only connect to one side of $A,B$ respectively, the graph construction forms a bipartite graph. To make the diameter distinction still work, we have to extend the paths in $T$ such that the distance from $t_i$ to $t_i'$ is at least 4 (this makes the diameter always be formed by a $s_i, t_j$-path and not some other path).
\end{proof}

\subsection{Permutation lower bounds}\label{sec:DiameterPerm}

In this section, we extend the list of our lower bounds by showing some reductions from $\textsc{Perm}_n$, which prove lower bounds for 1-pass algorithms, showing that they must use $\Omega(n\log n)$ bits. In particular, we show that there are constructions similar to the `Windmill' and `Diamond' constructions from the previous section that work for the \textsc{Permutation} problem.

\begin{figure}[!htp]
    \centering
    \includegraphics[width=.6\textwidth]{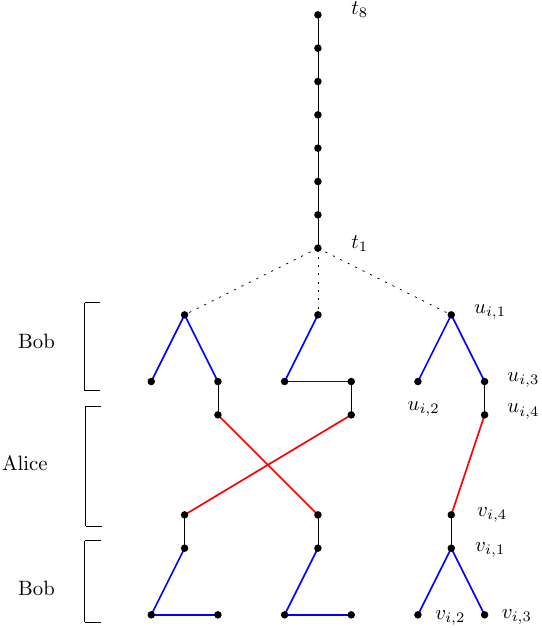}
    \caption{Construction we call `Windmill-Perm'. First we have a tail $t_1 \ldots t_8$. On the top Bob places the 1-construction only for the $\psi$-th index (blue edges, top). On the bottom Bob places the 1-construction when the $\gamma$-th bit of $i$ is 1 (blue edges, bottom). Alice attaches $u_{i,4}$ to $v_{\pi(i),4}$ (red edges, middle). The dotted segments are not edges, this is for ease of representation only, i.e. $u_{i,1} = t_1$ for every $i$.}
    \label{fig:perm_windmill}
\end{figure}

\begin{theorem}\label{thm:WindmillPerm}
    Any streaming algorithm for \textsc{Diameter} that works on the `Windmill-Perm' construction in the AL model using 1 pass over the stream requires $\Omega(n \log n)$ bits of memory.
\end{theorem}
\begin{proof}
    Let $\pi, j$ be the input to $\textsc{Perm}_n$, and let $\gamma$ and $\psi$ be the values associated with $j$ (see the definition of $\textsc{Perm}_n$). We create a graph construction called `Windmill-Perm' on $7n + 8$ vertices, see Figure~\ref{fig:perm_windmill}. Start with a path $t = t_1, \ldots, t_8$ on 8 vertices. For each index $i$, let $u_{i,1}, u_{i,2}, u_{i,3}$ and $v_{i,1}, v_{i,2}, v_{i,3}$ be two triplets. If $i$ is the $\psi$-th index, Bob inserts the edges $(u_{i,1},u_{i,2})$ and $(u_{i,2},u_{i,3})$, otherwise, Bob inserts the edges $(u_{i,1},u_{i,2})$ and $(u_{i,1},u_{i,3})$. For the other triplet, if the $\gamma$-th bit of $i$ is a 1, Bob inserts the edges $(v_{i,1},v_{i,2})$ and $(v_{i,2},v_{i,3})$, and otherwise Bob inserts the edges $(v_{i,1},v_{i,2})$ and $(v_{i,1},v_{i,3})$. Now add two vertices $u_{i,4}$ and $v_{i,4}$ and the edges $(u_{i,3}, u_{i,4})$ and $(v_{i,4}, v_{i,1})$. Now Alice inserts edges depending on her permutation $\pi$. For each $i$, Alice inserts the edge $(u_{i,4}, v_{\pi(i),4})$, i.e. the first triplet of index $i$ gets connected to the second triplet of index $\pi(i)$. We complete the construction by `glueing' the vertices $u_{i,1}$ onto $t_1$, these are the same vertex. Notice that this is a connected graph because a permutation is a bijective function. Given a streaming algorithm that works on such a graph, Alice and Bob can construct an AL stream corresponding to this graph by having Alice first reveal the vertices $u_{i,4},v_{i,4}$ for each $i$ (and their edges, which Alice knows), and then passing the memory state of the algorithm to Bob who reveals the rest of the vertices and their edges. This completes one pass of the stream.
    
    We now claim that the diameter of this graph is at least 14 when the answer to $\textsc{Perm}_n$ is YES, and otherwise the diameter is at most 13.
    
    Assume the answer to $\textsc{Perm}_n$ is YES, that is, the $\gamma$-th bit of the image of the $\psi$-th index under $\pi$ is 1. The triplet $u_{\psi,1}, u_{\psi,2}, u_{\psi,3}$ must have the edges $(u_{\psi,1},u_{\psi,2})$ and $(u_{\psi,2},u_{\psi,3})$ by construction. Then the edge $(u_{\psi,4}, u_{\pi(\psi),4})$ is present because of Alice, and this leads to a triplet with the edges $(v_{\pi(\psi),1},v_{\pi(\psi),2})$ and $(v_{\pi(\psi),2},v_{\pi(\psi),3})$ by the assumption that this is a YES instance. Therefore, the distance from $v_{\pi(\psi),3}$ to $t_8$ is 14, and so the diameter is at least 14.
    
    Now assume that the answer to $\textsc{Perm}_n$ is NO, that is, the $\gamma$-th bit of the image of the $\psi$-th index under $\pi$ is not 1. As the $\psi$-th index is the only $u$-triplet that forms a path and not a tree, and by the assumption the $v$-triplet it is connected to does not have the shape of a path, the distance from all vertices (except $t_8$) to $t_1$ is at most 6. Hence, the diameter is at most 13 because $t_8$ lies at distance 7 from $t_1$.
    
    We conclude that any 1-pass algorithm in the AL model that can solve the \textsc{Diameter} problem on the `Windmill-Perm' construction, must use $\Omega(n \log n)$ bits of memory by Proposition~\ref{prop:PermReduction}.
\end{proof}

\begin{corollary}\label{cor:WindmillPermResults}
    Any streaming algorithm for \textsc{Diameter} in the AL model that uses 1 pass over the stream must use $\Omega(n \log n)$ bits of memory, even on graphs for which the algorithm is given
    \begin{enumerate}[noitemsep]
        \item that the input is a bounded depth tree,
        \item that the Maximum Degree is a constant of at least 3.
    \end{enumerate}
\end{corollary}
\begin{proof}
    The corollary follows from Theorem~\ref{thm:WindmillPerm} together with observing that `Windmill-Perm' is
    \begin{enumerate}[noitemsep]
        \item a tree of constant depth,
        \item a lower bound construction that still works if we extend $t_1$ to a binary tree and extend the tail $t$ to consist of $8 + \log n$ vertices (this makes the diameter distinction to be between $13 + 2\log n$ and $14 + 2\log n$).
    \end{enumerate}
\end{proof}

Next we show a similar adaptation of the `Diamond' construction.

\begin{figure}[!htp]
    \centering
    \includegraphics[width=.7\textwidth]{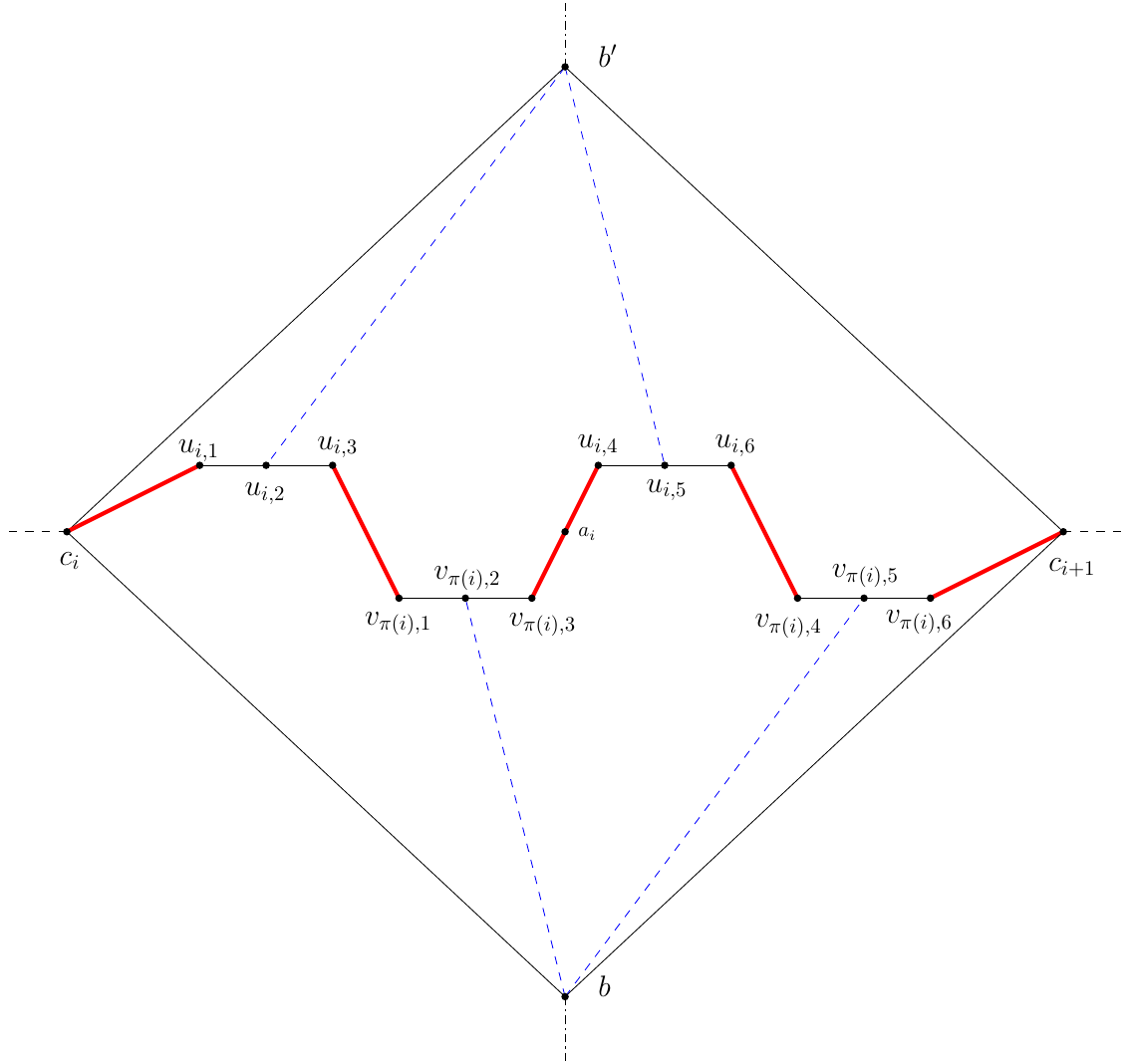}
    \caption{An adaptation for permutation of the `Diamond' construction, called `Diamond-Perm'. Blue dashed edges are present in the 0-case for Bob. Alice attaches top and bottom halves to each other with the bold red edges, depending on $\pi$. The vertices $b$ and $b'$ are connected with an edge, as indicated by the dashed-dotted lines.}
    \label{fig:perm_diamond}
\end{figure}

\begin{theorem}\label{thm:DiamondPerm}
    Any streaming algorithm for \textsc{Diameter} that works on the `Diamond-Perm' construction in the AL model using 1 pass over the stream requires $\Omega(n \log n)$ bits of memory.
\end{theorem}
\begin{proof}
    Let $\pi, j$ be the input to $\textsc{Perm}_n$, and let $\gamma$ and $\psi$ be the values associated with $j$ (see the definition of $\textsc{Perm}_n$). We create a graph construction called `Diamond-Perm' on $14n + 3$ vertices, see Figure~\ref{fig:perm_diamond}. The construction starts of with two vertices $b,b'$ connected with an edge. We then add $n+1$ vertices $c_0,\ldots,c_n$ each of which we connect to both $b$ and $b'$. For each index $i$ we do the following. We create four (disjoint) paths $(u_{i,1}, u_{i,2}, u_{i,3})$, $(u_{i,4}, u_{i,5}, u_{i,6})$, $(v_{i,1}, v_{i,2}, v_{i,3})$, and $(v_{i,4}, v_{i,5}, v_{i,6})$. If $i$ is not the $\psi$-th index, Bob inserts the edges $(b',u_{i,2})$ and $(b',u_{i,5})$. If the $\gamma$-th bit of $i$ is a 0, Bob inserts the edges $(b,v_{i,2})$ and $(b,v_{i,5})$. The rest of the edges depend on the permutation of Alice. For each $i$, Alice creates the vertex $a_i$ and inserts the edges $(c_{i-1}, u_{i,1})$, $(u_{i,3}, v_{\pi(i),1})$, $(v_{\pi(i),3}, a_i)$, $(a_i, u_{i,4})$, $(u_{i,6}, v_{\pi(i),4})$, and $(v_{\pi(i),6}, c_{i+1})$. Notice that, ignoring $b,b'$, this graph forms a path, as a permutation is a bijective function.  This completes the construction. Given a streaming algorithm that works on this construction, Alice and Bob can produce an AL stream by having Alice reveal the vertices $c_0, \ldots, c_n$, and $a_i, u_{i,1}, u_{i,3}, v_{i,1}, v_{i,3}, u_{i,4}, u_{i,6}, v_{i,4}, v_{i,6}$ for every $1 \leq i \leq n$. Then she passes the memory state of the algorithm to Bob who reveals the rest of the vertices including $b,b'$. This completes one pass of the stream.
    
    We now claim that the diameter of this graph is at least 10 when the answer to $\textsc{Perm}_n$ is YES, and otherwise the diameter is at most 9.
    
    Assume the answer to $\textsc{Perm}_n$ is YES, that is, the $\gamma$-th bit of the image of the $\psi$-th index under $\pi$ is 1. Consider $a_\psi$. The edges $(b',u_{\psi,2})$ and $(b',u_{\psi,5})$ are not present by construction. Also, the edges $(b,v_{\pi(\psi),2})$ and $(b,v_{\pi(\psi),5})$ are not present by construction. Therefore, the distance from $a_\psi$ to $b$ and $b'$ is 8. Then, the diameter is at least 10 because there are always vertices at distance at least 2 from both $b$ and $b'$.
    
    Now assume that the answer to $\textsc{Perm}_n$ is NO, that is, the $\gamma$-th bit of the image of the $\psi$-th index under $\pi$ is not 1. Then, for every $1\leq i\leq n$ between $c_{i-1}$ and $c_i$ either or both of the edge pairs $(b',u_{i,2}),(b',u_{i,5})$ and $(b,v_{\pi(i),2}), (b,v_{\pi(i),5})$ must be present, because the edges $(b',u_{i,2}),(b',u_{i,5})$ are only absent for the $\psi$-th index, and the $\gamma$-th bit under the image of $\pi$ is not a 1. Hence, the distance from any vertex to one of $b$ or $b'$ is at most 4. As $b$ and $b'$ are connected with an edge, the diameter is at most 9.
    
    We conclude that any 1-pass algorithm in the AL model that can solve the \textsc{Diameter} problem on the `Diamond-Perm' construction, must use $\Omega(n \log n)$ bits of memory by Proposition~\ref{prop:PermReduction}.
\end{proof}

\begin{corollary}\label{thm:DiamondPermResults}
    Any streaming algorithm for \textsc{Diameter} in the AL model that uses 1 pass over the stream must use $\Omega(n \log n)$ bits of memory, even on graphs for which the algorithm is given a Deletion Set to a path of size at least 2.
\end{corollary}
\begin{proof}
    The corollary follows from Theorem~\ref{thm:DiamondPerm}, together with observing that `Diamond-Perm' is a path when we remove $\{b, b'\}$.
\end{proof}


\section{Connectivity}\label{sec:Connectivity}

In this section, we show results for \textsc{Connectivity}. \textsc{Connectivity} is an easier problem than \textsc{Diameter}, that is, solving \textsc{Diameter} solves \textsc{Connectivity} as well, but not the other way around. Hence, lower bounds in this section also imply lower bounds for \textsc{Diameter} (in non-connected graphs).
In general graphs, a single pass, $\O(n \log n)$ bits of memory algorithm exists by maintaining connected components in a Disjoint Set data structure~\cite{McGregorSurveyStream}, which is optimal in general graphs~\cite{SunWoodruffPermBounds}.
The interesting part about \textsc{Connectivity} is that some graph classes admit fairly trivial algorithms by a counting argument. For example, if the input is a forest, we can decide on \textsc{Connectivity} by counting the number of edges, which is a 1-pass, $\O(\log n)$ bits of memory, algorithm. An overview of the results in this section is given in Table~\ref{table:Connectivity}.

\begin{table}[ht]
\begin{tabular}{|lr|l|l|l|}
\hline
 Parameter ($k$) / Graph class & Size & Bound & Thm\\ \hline
 General Graphs & & \cellcolor{LightGreen}$(\textsc{EA}, n \log n, 1)$-str. by Disjoint Set data structure~\cite{McGregorSurveyStream} & - \\ 
                & & \cellcolor{LightPink}$(\textsc{EA}, n \log n, 1)$-hard by Sun and Woodruff~\cite{SunWoodruffPermBounds} & - \\\hline
 \textsc{Vertex Cover Number} & $\geq 0$ & \cellcolor{LightGreen}$(\textsc{AL},k\log n,1)$-str.: Disjoint Set on Vertex Cover & \ref{obs:VCcon} \\
 & $\geq 2$ & \cellcolor{LightPink}$(\textsc{VA}, n/p, p)$-hard by Henzinger et al.~\cite{HenzingerStreams} & - \\ \hline
 \textsc{Distance to $\ell$ cliques} & $\geq 0$ & \cellcolor{LightGreen}$(\textsc{AL},(k+\ell)\log n,1)$-str.: Disjoint Set data structure & \ref{obs:DistkCliquecon} \\ \hline
 \textsc{FVS} & $= 0$ & \cellcolor{LightGreen}$(\textsc{EA},\log n,1)$-str. by counting. & - \\
 & $\geq 1$ & \cellcolor{LightPink}$(\textsc{AL}, n/p, p)$-hard & \ref{cor:SimpleAL-ConnResults}\\ \hline
 \textsc{FES} & $\geq 0$ & \cellcolor{LightGreen}$(\textsc{EA},\log n,1)$-str. by counting. & - \\ \hline
 \textsc{Distance to matching} & $\geq 2$ & \cellcolor{LightPink}$(\textsc{AL}, n/p, p)$-hard & \ref{cor:SimpleAL-ConnResults}\\ \hline
 \textsc{Distance to path} & $\geq 0$ & \cellcolor{LightGreen}$(\textsc{EA},k \log n,1)$-str. by checking connection to path & - \\ \hline
 \textsc{Distance to depth $\ell$ tree} & $\geq 0$ & \cellcolor{LightGreen}$(\textsc{EA},k \log n,1)$-str. by checking connection to tree & - \\ \hline
 \textsc{Domination Number} & $\geq 2$ & \cellcolor{LightPink}$(\textsc{AL}, n/p, p)$-hard & \ref{cor:SimpleAL-ConnResults}\\ \hline
 \textsc{Distance to Chordal} & $\geq 1$ & \cellcolor{LightPink}$(\textsc{AL}, n/p, p)$-hard & \ref{cor:SimpleAL-ConnResults}\\ \hline
 \textsc{Maximum Degree} & $\geq 2$ & \cellcolor{LightPink}$(\textsc{AL}, n/p, p)$-hard, $(\textsc{AL}, n \log n, 1)$-hard & \ref{thm:Cycles-Conn}, \ref{thm:CyclesPermConn}\\\hline
 Bipartite Graphs & & \cellcolor{LightPink}$(\textsc{AL}, n/p, p)$-hard, $(\textsc{AL}, n \log n, 1)$-hard & \ref{thm:Cycles-Conn}, \ref{thm:CyclesPermConn}\\ \hline
 Interval Graphs & & \cellcolor{LightPink}$(\textsc{VA}, n/p, p)$-hard & \ref{thm:Interval-Conn}\\\hline
 Split graphs & & \cellcolor{LightGreen}$(\textsc{EA}, n/p, p)$-str. by finding degree 0 vertex & - \\
 & & \cellcolor{LightPink}$(\textsc{VA}, n/p, p)$-hard & \ref{thm:Split-Conn}\\ \hline
\end{tabular}
\caption{Overview of the results for \textsc{Connectivity}. All hardness results listed here are through reductions from \textsc{Disjointness}. $(\M, m, p)$-hard means that any algorithm using $p$ passes in model $\M$ (or weaker) requires $\Omega(m)$ bits of memory. $(\M, m, p)$-str. means that there is an algorithm that uses $p$ passes in model $\M$ (or stronger) using $\O(m)$ bits of memory. \textsc{FVS} stands for Feedback Vertex Set number, \textsc{FEN} for Feedback Edge Set number. We state most upper bounds only as observations.}
\label{table:Connectivity}
\end{table}

The following upper bounds follow from applications of the Disjoint Set data structure.

\begin{restatable}{observation}{VCconObs}\label{obs:VCcon}
    Given a graph $G$ as an \textsc{AL} stream with vertex cover number $k$, we can solve \textsc{Connectivity $[k]$} in $1$ pass and $\O(k \log n)$ bits of memory.
\end{restatable}
\begin{proof}
    When the vertex cover is known, we can keep track of a Disjoint Set data structure on the $k$ vertices of the vertex cover. Seeing any vertex that connects two or more vertices of the vertex cover in the stream translates directly to taking the union of the corresponding sets in the data structure. If at the end of the stream the data structure contains only one set and we have not seen a degree-0 vertex, the graph is connected.
    
    When the vertex cover is not given, and only its size, we can greedily maintain an approximate vertex cover of size at most $2k$ by maintaining a maximal matching, while executing the above procedure on this set. 
\end{proof}

\begin{restatable}{observation}{CliqueConObs}\label{obs:DistkCliquecon}
    Given a graph $G$ as an \textsc{AL} stream with a deletion set $X$ of size $k$ to $\ell$ cliques, we can solve \textsc{Connectivity $[k, \ell]$} in $1$ pass and $\O((k + \ell) \log n)$ bits of memory.
\end{restatable}
\begin{proof}
    We use a Disjoint Set data structure on all vertices in $X$ and a representative vertex for each clique, say the lowest numbered vertex of that clique. The space used by the data structure is $\O((k+\ell)\log n)$ bits. For a vertex in $X$ we only process the edges to other vertices in $X$ in the data structure. For a vertex in a clique ($\notin X$) we register the edges to vertices in $X$ and its lowest number neighbour $\notin X$. This takes at most $(k+1)\log n$ bits, and is enough to take unions in the data structure corresponding to the connections seen by this vertex.
    
    At the end of the stream, if the data structure contains only one set, the graph is connected.
\end{proof}

Next is a simple lower bound for the AL model.

\begin{figure}
  \centering
  \includegraphics[width=.6\linewidth]{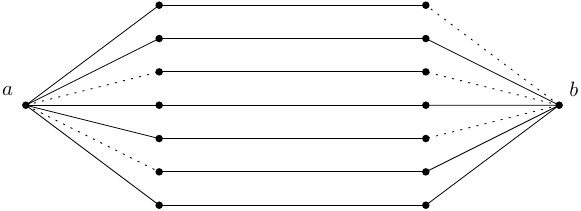}
  \captionof{figure}{AL lower bound for connectivity, called `Simple AL-Conn'. On contrary to the  similar diameter lower bound, edges are present when the corresponding entry is a 0.}
  \label{fig:SimpleAL-Conn}
\end{figure}

\begin{restatable}{theorem}{SimpleALConn}\label{thm:SimpleAL-Conn}
    Any streaming algorithm for \textsc{Connectivity} that works on the `Simple AL-Conn' construction in the AL model using $p$ passes over the stream requires $\Omega(n/p)$ bits of memory.
\end{restatable}
\begin{proof}
    Let $x,y$ be the input to $\textsc{Disj}_n$ of Alice and Bob, respectively. We construct a graph as illustrated in Figure~\ref{fig:SimpleAL-Conn}. Let $M$ be a matching on $n+1$ edges, and associate each edge with an index $1\leq i \leq n$. Now we add two vertices $a,b$ of Alice and Bob respectively, and connect $a$ and $b$ to the $(n+1)$-th edge. The edge from $a$ to the $i$-th edge in $M$ is present if and only if $x_i = 0$, for all $1 \leq i \leq n$. The same happens for $b$, where the edge from $b$ to the $i$-th edge in $M$ is present if and only if $y_i = 0$, for all $1 \leq i \leq n$. This completes the construction, it has $2n+4$ vertices. Given a streaming algorithm that works on a family including this construction, Alice and Bob construct the AL stream as follows. First, Alice reveals $a$ and the vertices on the left of $M$. Then she passes the memory state of the algorithm to Bob who reveals $b$ and the vertices on the right of $M$, which completes one pass of the stream.
    
    We now claim that the graph is connected if and only if the answer to $\textsc{Disj}_n$ is YES.
    
    Let us assume that the answer to $\textsc{Disj}_n$ is NO, that is, there is an index $i$ such that $x_i = y_i = 1$. Then clearly the $i$-th edge is not connected to the rest of the graph (which includes $a$ and $b$).
    
    Now assume the answer to $\textsc{Disj}_n$ is YES, that is, there is no index $i$ such that $x_i = y_i = 1$. Notice that there is always a path from $a$ to $b$ via the $(n+1)$-th edge of $M$. Furthermore, by the assumption, for each index $1\leq i \leq n$ either Alice or Bob (or both) has a 0, which means the $i$-th edge is connected to either $a$ or $b$. Hence, the graph is connected.
    
    We conclude that any algorithm that can solve the \textsc{Connectivity} problem on the graph construction `Simple AL-Conn' in the AL model in $p$ passes, must use $\Omega(n/p)$ bits of memory by Proposition~\ref{prop:DisjReduction}.
\end{proof}

The following follows from Theorem~\ref{thm:SimpleAL-Conn} by observing properties of the `Simple AL-Conn' construction.

\begin{restatable}{corollary}{SimpleALConnResults}\label{cor:SimpleAL-ConnResults}
    Any streaming algorithm for \textsc{Connectivity} in the AL model that uses $p$ passes over the stream must use $\Omega(n/p)$ bits of memory, even on graphs for which the algorithm is given a
    \begin{itemize}[noitemsep]
        \item Feedback Vertex Set of size at least 1,
        \item Deletion Set to Matching of size at least 2,
        \item Dominating Set of size at least 2.
    \end{itemize}
\end{restatable}
\begin{proof}
    The corollary follows from Theorem~\ref{thm:SimpleAL-Conn}, together with observing that the construction of `Simple AL-Conn' is
    \begin{itemize}[noitemsep]
        \item a forest when removing $\{a\}$,
        \item a matching when removing $\{a, b\}$,
        \item dominated by the set of vertices $\{a,b\}$.
    \end{itemize}
\end{proof}

An interesting lower bound is for a unique case: graphs of maximum degree 2. We mentioned that for a forest we have a simple counting algorithm for \textsc{Connectivity}, so the hardness must be for some graph which consists of one or more cycles. Although Theorem~\ref{thm:SimpleAL-Conn} implies \textsc{Connectivity} is hard for graphs with a Feedback Vertex Set of constant size, we now show that in the specific case of maximum degree 2-graphs, the problem is still hard, see Figure~\ref{fig:Cycles-Conn} for an illustration of the construction. We note that this reduction is similar to the problem tackled by Verbin and Yu~\cite{YuVerbinCycleCounting} and Assadi et al.~\cite{AssadiKSY20}, but our result is slightly stronger in this setting, as it concerns a distinction between 1 or 2 disjoint cycles.

\begin{figure}[th]
    \centering
    \includegraphics[width=\textwidth]{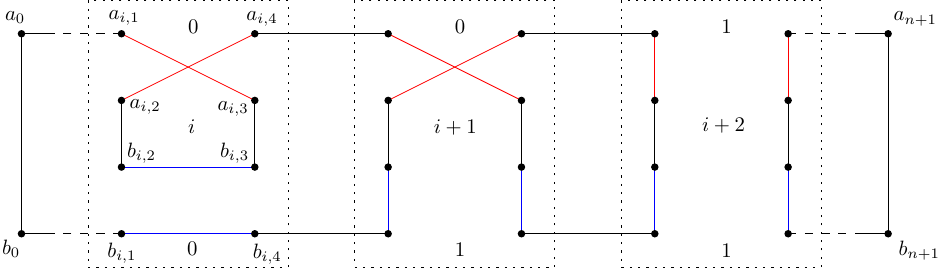}
    \caption{AL lower bound for connectivity, called `Cycles'. The graph consists of one or multiple cycles depending on the output of \textsc{Disj}$_n$. The black edges are always present. The red (blue) edges between $a$-vertices ($b$-vertices) are controlled by Alice (Bob) and are in a crossing (horizontal) or vertical configuration depending on whether the $i$-th entry of Alice (Bob) is $0$ or $1$.}
    \label{fig:Cycles-Conn}
\end{figure}

\begin{restatable}{theorem}{CyclesConn}\label{thm:Cycles-Conn}
    Any streaming algorithm for \textsc{Connectivity} that works on graphs of maximum degree 2 in the AL model using $p$ passes over the stream requires $\Omega(n/p)$ bits of memory.
\end{restatable}
\begin{proof}
    Let $x,y$ be the input to $\textsc{Disj}_n$ of Alice and Bob, respectively. We create a construction as shown in Figure~\ref{fig:Cycles-Conn}. We create a graph on $8n + 4$ vertices. Associate 8 vertices $a_{i,1}, \ldots, a_{i,4}$ and $b_{i,1}, \ldots, b_{i,4}$ with each index $i$. Let us call the remaining four vertices $a_0, b_0, a_{n+1}, b_{n+1}$, and insert the edges $(a_0,b_0)$ and $(a_{n+1}, b_{n+1})$. Then for each index $i$, we do the following. Insert the edges $(a_{i,2}, b_{i,2})$, $(a_{i,3},b_{i,3})$, $(a_{i,1}, a_{i-1,4})$, $(b_{i,1}, b_{i-1,4})$, $(a_{i,4}, a_{i+1,1})$, $(b_{i,4}, b_{i+1,1})$, where $a_{i-1,4}, b_{i-1,4}$ are replaced with $a_0,b_0$ when $i=1$, and $a_{i+1,1}, b_{i+1,1}$ are replaced with $a_{n+1}, b_{n+1}$ when $i=n$. These are all the fixed edges. For each $i$, Alice also inserts $(a_{i,1},a_{i,3})$ and $(a_{i,2}, a_{i,4})$ when $x_i = 0$ or inserts $(a_{i,1},a_{i,2})$ and $(a_{i,3}, a_{i,4})$ when $x_i = 1$. Bob inserts $(b_{i,1}, b_{i,4})$ and $(b_{i,2}, b_{i,3})$ when $y_i = 0$, or inserts $(b_{i,1}, b_{i,2})$ and $(b_{i,3}, b_{i,4})$ when $y_i = 1$. This completes the construction. Given an algorithm that works on a family including this construction, Alice and Bob construct an AL stream as follows. First, Alice reveals the vertices $a_0, a_{n+1}$ and $a_{i,k}$ for all $1\leq i\leq n, 1\leq k\leq 4$, then passes the memory state to Bob who reveals the vertices $b_0, b_{n+1}$ and $b_{i,k}$ for all $1\leq i\leq n, 1\leq k\leq 4$. This completes one pass of the stream. Notice that Alice and Bob do not need information about the input of the other to do this, as there are only fixed edges between $a$- and $b$-vertices. Also notice that this graph always consists of (a disjoint union of) one or more cycles regardless of the input to $\textsc{Disj}_n$, as every vertex in the graph has degree 2.
    
    We now claim that the graph is connected if and only if the answer to $\textsc{Disj}_n$ is YES.
    
    Let us assume that the answer to $\textsc{Disj}_n$ is NO, that is, there is an index $i$ such that $x_i = y_i = 1$. It is easy to see that there is no path between $a_{i,1}$ and either $a_{i,4}$ or $b_{i,4}$, and similarly, there is no path between $b_{i,1}$ and either $a_{i,4}$ or $b_{i,4}$. Hence the graph is not connected.
    
    Now assume the answer to $\textsc{Disj}_n$ is YES, that is, there is no index $i$ such that $x_i = y_i = 1$. We will construct a simple path from $a_0$ to either $a_{n+1}$ or $b_{n+1}$. If this succeeds, then the graph must be a single cycle, as we can continue to the other of $a_{n+1}$ or $b_{n+1}$ and walk the other way to $b_0$, never crossing the first path because every vertex has degree 2. These two paths together with the edges $(a_0, b_0)$, $(a_{n+1}, b_{n+1})$ form a single cycle. Starting at $a_0$, we can view a path going `right', crossing each index $i$ step by step. At an $a_{i,1}$, there are only two possible cases: either we walk through $a_{i,2}, a_{i,3}, b_{i,2}, b_{i,3}$ in some order and end in $a_{i,4}$, or we have a path to $b_{i,4}$ (using only vertices of index $i$). In both cases, we can advance to the next $i$. At an $b_{i,1}$, there are also only two cases: either there is an edge to $b_{i,4}$ or there is a path through $b_{i,2}, a_{i,2}$ to $a_{i,4}$. In both cases, we can advance to the next $i$. Hence, we can find a path walking through each $i$ advancing to the next, which must mean we end up in either $a_{n+1}$ or $b_{n+1}$, and we are done.
    
    We conclude that any algorithm that can solve the \textsc{Connectivity} problem on the graph construction `Cycles' in the AL model in $p$ passes, must use $\Omega(n/p)$ bits of memory by Proposition~\ref{prop:DisjReduction}.
\end{proof}

We note that we can make the result of Theorem~\ref{thm:Cycles-Conn} (and \ref{thm:CyclesPermConn}) hold for bipartite graphs of maximum degree 2 by subdividing every edge, making the graph odd cycle-free, and thus bipartite.

The proofs of Theorems~\ref{thm:HfreeOverviewConn} and \ref{thm:DistHfreeOverviewConn} follow.
\HfreeOverviewConn*
\begin{proof}
If $H$ contains a cycle of length not equal to $6$ as a subgraph, then the result follows from Theorem~\ref{thm:SimpleAL-Conn}, because that construction is $C_\ell$-free for any $\ell \not= 6$. By subdividing the middle (matching) edges, the construction can be made $C_\ell$-free for any fixed $\ell > 2$. Hence, we may assume that $H$ does not contain a cycle and thus is a forest. If $H$ contains a vertex of degree at least~$3$, then the result follows from Theorem~\ref{thm:Cycles-Conn}, because that construction has maximum degree~$2$. Hence, we may assume that $H$ is a linear forest.
If $H$ contains a $P_7$, then the result follows from a slight adaptation of the construction of Theorem~\ref{thm:SimpleAL-Conn}. By making the vertices $a$ and $b$ of that construction adjacent, the resulting graph cannot have a $P_7$ as an induced subgraph, while not affecting the correctness of the construction.
\end{proof}

\DistHfreeOverviewConn*
\begin{proof}
This is an immediate corollary of Theorem~\ref{thm:SimpleAL-Conn}. In that construction, after removing vertices $a$ and $b$, the remainder is a disjoint union of $P_2$'s.
\end{proof}

Interval and split graphs are hard in the VA model, see Figures~\ref{fig:IntervalConn} and \ref{fig:SplitConn}.

\begin{figure}
\centering
\begin{minipage}{.58\textwidth}
  \centering
  \includegraphics[width=\linewidth]{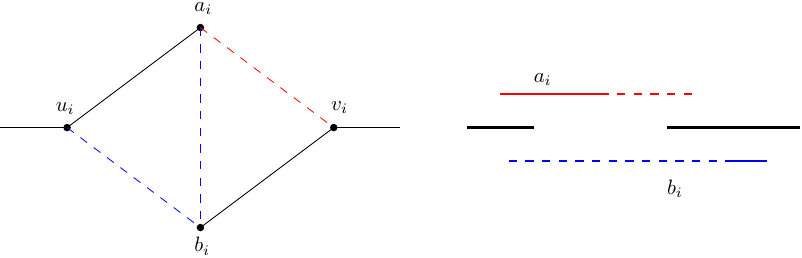}
  \captionof{figure}{VA lower bound for connectivity on interval graphs, called `Interval'. We see the gadget for index $i$, where the dotted lines are present when the corresponding value is 0. In essence, the intervals $a_i$ and $b_i$ are short or long depending on the input of Alice and Bob. The $n$ gadgets are placed consecutively.}
  \label{fig:IntervalConn}
\end{minipage}\quad
\begin{minipage}{.38\textwidth}
  \centering
  \includegraphics[width=.55\linewidth]{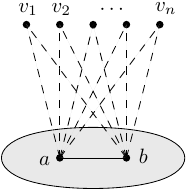}
  \captionof{figure}{VA lower bound for connectivity on split graphs, called `Split-Conn'. The dashed edges towards $v_i$ are present when there is a 0 at index $i$.
  }
  \label{fig:SplitConn}
\end{minipage}
\end{figure}

\begin{restatable}{theorem}{IntervalConn}\label{thm:Interval-Conn}
    Any streaming algorithm for \textsc{Connectivity} that works on interval graphs in the VA model using $p$ passes over the stream requires $\Omega(n/p)$ bits of memory.
\end{restatable}
\begin{proof}
    Let $x,y$ be the input to $\textsc{Disj}_n$ of Alice and Bob, respectively. We create a construction as shown in Figure~\ref{fig:IntervalConn}. We create an interval graph on $4n$ vertices. For each $i$, we create the vertices $u_i, v_i, a_i, b_i$. We insert the edges $(u_i, a_i)$, $(v_i, b_i)$, and $(v_{i-1},u_i)$ (for $i=1$ we do not insert this last edge). Alice inserts the edge $(a_i, v_i)$ if and only if $x_i = 0$. Bob inserts the edges $(b_i, u_i)$ and $(b_i, a_i)$ if and only if $y_i = 0$. This completes the construction. Notice that this is an interval graph, as illustrated by the interval representation of index $i$ in Figure~\ref{fig:IntervalConn}. Given an algorithm that works on a family including this construction, Alice and Bob construct an VA stream as follows. First, Alice reveals all vertices $u_i, a_i, v_i$ (and the edges between them) for each $i$. Then she passes the memory of the algorithm to Bob who reveals each $b_i$. This completes one pass of the stream. Notice that Alice does not need to know the input of Bob for $\textsc{Disj}_n$, and neither does Bob have to know the input of Alice, as it is a VA stream.
    
    We now claim that the graph is connected if and only if the answer to $\textsc{Disj}_n$ is YES.
    
    Let us assume that the answer to $\textsc{Disj}_n$ is NO, that is, there is an index $i$ such that $x_i = y_i = 1$. Then clearly, there is no path between $u_i$ and $v_i$, and so the graph is not connected.
    
    Now assume the answer to $\textsc{Disj}_n$ is YES, that is, there is no index $i$ such that $x_i = y_i = 1$. Now we claim that there is a path from $u_1$ to $v_{n}$ using every $u_i$ and $v_i$, and hence the graph is connected. Indeed, if there is such a path then the graph is connected, as each $a_i$ and $b_i$ are always adjacent to $u_i$ and $v_i$, respectively. There is such a path because, for each $1 \leq i \leq n$ at least one of the edges $(u_i, b_i)$ or $(a_i, v_i)$ is present, creating a path between $u_i$ and $v_i$. Combining these paths for each $i$ gives us the path we were looking for, as the edges $(v_{i-1},u_i)$ exist for each $2\leq i\leq n$.
    
    We conclude that any algorithm that can solve the \textsc{Connectivity} problem on the graph construction `Interval' in the VA model in $p$ passes, must use $\Omega(n/p)$ bits of memory by Proposition~\ref{prop:DisjReduction}.
\end{proof}

\begin{restatable}{theorem}{SplitConn}\label{thm:Split-Conn}
    Any streaming algorithm for \textsc{Connectivity} that works on split graphs in the VA model using $p$ passes over the stream requires $\Omega(n/p)$ bits of memory.
\end{restatable}
\begin{proof}
    Let $x,y$ be the input to $\textsc{Disj}_n$ of Alice and Bob, respectively. We create a construction as shown in Figure~\ref{fig:SplitConn}. We create an split graph on $n + 2$ vertices. Let $v_1,\ldots, v_n$ be $n$ vertices in the independent set. Let $a,b$ be two vertices that form the clique. Alice inserts the edges $(a,v_i)$ when $x_i = 0$, and Bob inserts the edges $(b,v_i)$ when $y_i = 0$, for each $1\leq i \leq n$. This completes the construction. Given an algorithm that works on the construction, Alice and Bob construct a VA stream as follows. First, Alice reveals $v_1,\ldots, v_n$ (without edges at this point) and then reveals $a$. She then passes the memory of the algorithm to Bob, who reveals $b$, which completes one pass of the stream.
    
    It can be easily seen that there is an isolated vertex if and only if there is an index $i$ such that $x_i = y_i = 1$. Split graphs are connected if and only if there is no isolated vertex.
    
    We conclude that any algorithm that can solve the \textsc{Connectivity} problem on split graphs in the VA model in $p$ passes, must use $\Omega(n/p)$ bits of memory by Proposition~\ref{prop:DisjReduction}.
\end{proof}

For split graphs, in any model, \textsc{Connectivity} admits a one-pass, $\O(n)$ bits of memory algorithm by counting if there is a vertex of degree 0 (and so also for any $p$ a $p$-pass algorithm using $\O(n/p)$ bits by splitting up the work in $p$ parts)\footnote{This assumes the vertices are labelled $1 \ldots n$ and do not have arbitrary labels.}. If there can be no isolated vertices, then a split graph is always connected.

\subsection{Permutation Lower Bound}\label{sec:ConnectivityPerm}

As an additional result, we give another reduction for graphs of maximum degree 2. Sun and Woodruff~\cite{SunWoodruffPermBounds} have already shown a \textsc{Permutation} lower bound for \textsc{Connectivity} in general graphs. We show that the `Cycles' construction can be extended to lower bounds using the $\textsc{Perm}_n$ problem, showing hardness for graphs of maximum degree 2.

\begin{figure}
    \centering
    \includegraphics[width=\textwidth]{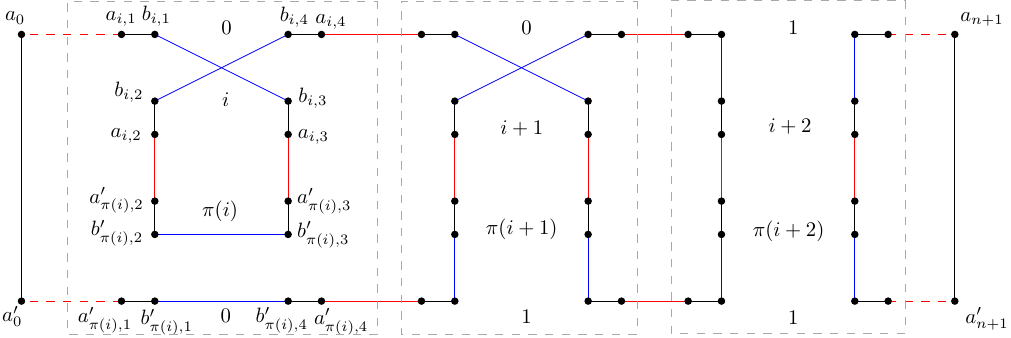}
    \caption{AL permutation lower bound for connectivity, called `Cycles-Perm'. The graph consists of one or multiple cycles depending on the input of $\textsc{Perm}_n$. Red edges between $a$-vertices are controlled by Alice, blue edges between $b$-vertices by Bob.}
    \label{fig:Cycles-Perm-Conn}
\end{figure}

\begin{theorem}\label{thm:CyclesPermConn}
    Any streaming algorithm for \textsc{Connectivity} that works on graphs of maximum degree 2 in the AL model using 1 pass over the stream requires $\Omega(n \log n)$ bits of memory.
\end{theorem}
\begin{proof}
    Essentially, we adapt the construction of Theorem~\ref{thm:Cycles-Conn} to work for $\textsc{Perm}_n$. Let $\pi, j$ be the input to $\textsc{Perm}_n$, and let $\gamma$ and $\psi$ be the values associated with $j$ (see the definition of $\textsc{Perm}_n$). We construct a graph on $16n+4$ vertices, consisting of one or more cycles. For each index $i$, we create the vertices $a_{i,1}, a_{i,2}, a_{i,3}, a_{i,4}, a_{i,1}', a_{i,2}', a_{i,3}', a_{i,4}'$ and $b_{i,1}, b_{i,2}, b_{i,3}, b_{i,4}, b_{i,1}', b_{i,2}', b_{i,3}', b_{i,4}'$, where we connect the each $a$-vertex with its corresponding $b$-vertex (i.e. $(a_{i,2}', b_{i,2}')$ is an edge). Next to this, we also create 4 vertices $a_0, a_0', a_{n+1}, a_{n+1}$, where $(a_0, a_0')$ and $(a_{n+1}, a_{n+1}')$ are edges. The rest of the edges are dependent on the input to $\textsc{Perm}_n$. For an index $i$, if it is the $\psi$-th index, Bob inserts the edges $(b_{i,1}, b_{i,2})$ and $(b_{i,3}, b_{i,4})$. If it is not the $\psi$-th index, Bob inserts the edges $(b_{i,1}, b_{i,3})$ and $(b_{i,2}, b_{i,4})$. Also, for an index $i$, if its $\gamma$-th bit is a 1, Bob inserts the edges $(b_{i,1}', b_{i,2}')$ and $(b_{i,3}', b_{i,4}')$. If the $\gamma$-th bit is a 0, Bob inserts the edges $(b_{i,1}', b_{i,4}')$ and $(b_{i,2}', b_{i,3}')$. Alice links `the top' of $i$ with `the bottom' of $\pi(i)$. For each index $i$, Alice inserts the following edges, $(a_{i-1,4}, a_{i,1})$ (or $(a_0, a_{i,1})$ when $i=1$), $(a_{i,4}, a_{i+1,4})$ (or $(a_{i,4}, a_{n+1})$ when $i=n$), $(a_{i,2}, a_{\pi(i), 2}')$, $(a_{i,3}, a_{\pi(i), 3}')$, $(a_{\pi(i-1),4}', a_{\pi(i),1}')$ (or $(a_0', a_{\pi(i),1}')$ when $\pi(i)=1$), $(a_{\pi(i),4}', a_{\pi(i+1),4}')$ (or $(a_{\pi(i),4}', a_{n+1}')$ when $\pi(i)=n$).\footnote{We note that formally, we insert many edges twice, but this is to make the description more understandable. Alice does not actually insert these edges twice.} This concludes the construction. The graph consists of one or more cycles because every vertex has degree 2. Given an algorithm that works on a family including this construction, Alice and Bob construct an AL stream as follows. First, Alice reveals all $a$-vertices, then passes the memory of the algorithm to Bob, who reveals all $b$-vertices, which completes one pass of the stream. This is correct, as $a$-vertices are only connected to $b$-vertices with edges independent of the input to $\textsc{Disj}_n$.
    
    We claim that the graph is not connected if and only if the answer to $\textsc{Perm}_n$ is YES. The graph is not connected if and only if there is an index $i$ such that `the top' and `the bottom' are both a 1-construction. This can only be the case when an index $i$ is the $\psi$-th index on `the top', and the $\gamma$-th bit is a 1 on `the bottom'. However, `the bottom' corresponds to the index $\pi(i)$ because of the edges of Alice, which means that the $\gamma$-th bit of the image under $\pi$ of $\psi$ is a 1. Hence, this occurs if and only if the answer to $\textsc{Perm}_n$ is YES.
    
    We conclude that any 1-pass algorithm in the AL model that can solve the \textsc{Connectivity} problem on graphs of maximum degree 2, must use $\Omega(n \log n)$ bits of memory by Proposition~\ref{prop:PermReduction}.
\end{proof}

The results of Theorem~\ref{thm:CyclesPermConn} also holds for bipartite graphs of degree 2. To see this, subdivide every edge, making the graph odd cycle-free, and thus bipartite.


\section{Vertex Cover kernelization}\label{sec:VCkernel}

In this section\footnote{This section is based on the master thesis ``Parameterized Algorithms in a Streaming Setting'' by the first author.}, we parameterize the \textsc{Vertex Cover} problem by the solution size $k$. We now show how our insights into parameterized, streaming graph exploration can aid in producing a new kernelization algorithm for \textsc{Vertex Cover [$k$]}.
The basis for our result is a well-known kernel for the \textsc{Vertex Cover [$k$]} problem of Buss and Goldsmith~\cite{BussVCkernel}, consisting of $\O(k^2)$ edges. Constructing this kernel is simple: find all vertices with degree bigger than $k$, and remove them from the graph, and decrease the parameter with the number of vertices removed, say to $k'$. Then, there is no solution if there are more than $k\cdot k'$ edges. Therefore, we have a kernel consisting of $\O(k^2)$ edges. We are able to achieve this same kernel in the AL model, as counting the degree of a vertex is possible in this model. Interestingly, we do not require $\O(k^2\log n)$ bits of memory to produce a stream corresponding to the kernel of $\O(k^2)$ edges. This result is also possible in the EA model, by allowing vertices up to degree $2k$.

\begin{theorem}\label{thm:VCkernelk^2edgesEAAL}
    Given a graph $G$ as an AL stream, we can make an AL stream corresponding to an $\O(k^2)$-edge kernel for the \textsc{Vertex Cover [$k$]} problem using two passes and $\O(k\log n)$ bits of memory. When we work with an EA stream, we can make an EA stream corresponding to an $\O(k^2)$-edge kernel using four passes and $\O(k\log n)$ bits of memory.
\end{theorem}
\begin{proof}
    Let $G$ be a graph, with $n$ vertices and $m$ edges, given as an AL stream, and let $k$ be the solution size parameter for the \textsc{Vertex Cover [$k$]} problem. Note that we can count the degree of every vertex when it appears in the stream, as we are given all adjacencies of a vertex consecutively. Therefore, in one pass over the stream we can count the degree of every vertex, and save each vertex with a degree bigger than $k$ in a set $S$, as long as $|S| \leq k$. In this same pass, we keep track of two more counters: the total number of edges in the stream $m'$ (which is $2m$), and the number of unique edges we remove $r$. We find $r$ by incrementing a local (to a vertex) counter $r'$ when we see edges towards vertices not in $S$, and add $r'$ to $r$ if we decide to add the vertex to $S$. If $\frac{m'}{2} - r > k \cdot (k - |S|)$, return NO. Otherwise, make a pass over the stream, and output only those edges between vertices not in $S$.
	
	The output must be an AL stream, as we only remove edges from an AL stream to produce it. Let is also be clear that we use two passes over the stream.
	
	The set $S$ takes $\O(k\log n)$ bits of memory, as finding more than $k$ vertices will result in returning NO. Counting the total number of edges takes $\O(\log m) = \O(\log n)$ bits of space, and other constant number of counters are the same size or smaller. Therefore, this procedure uses $\O(k\log n)$ bits of memory.
	
	The behaviour of this procedure is equivalent of the kernelization algorithm of Buss and Goldsmith~\cite{BussVCkernel}, as it finds exactly those vertices with degree higher than $k$, and `removes' them by adding them to $S$ and ignoring edges incident to them in the output. Checking the instance size is done correctly, as the new parameter $k'$ is equivalent to $k - |S|$, and the number of remaining edges is equal to $m - r$, which is $\frac{m'}{2} - r$. The value of $r$ is counted correctly because we only count unique edges by ignoring those towards vertices already in $S$. Therefore, this kernelization procedure is correct.
	
	For the case of the EA stream, we essentially do the same as in the AL model, but in this model we end up with a slightly larger kernel. In our first pass, we greedily construct a vertex cover $X$ of size $2k$ or conclude that there is no solution to \textsc{Vertex Cover [$k$]}. Now, we know that all vertices not in $X$ must have degree at most $2k$. So, in the second pass we count the degree of all vertices in $X$, and if a degree exceeds $2k$ we the add vertex to $S$. This takes only $\O(k\log n)$ bits of memory. If at some point $|S| > k$ we also stop and conclude there is no solution to \textsc{Vertex Cover [$k$]}, as all vertices of degree at least $2k$ must be in any solution. In a third pass, we count the total number $m$ of edges in the stream, and the number of unique edges we remove $r$ (this is not possible during the second pass because we might count edges double when it has both endpoints in $S$). Now if $m - r > 2k \cdot (k - |S|)$, we can conclude there is no solution to \textsc{Vertex Cover [$k$]}, as each vertex has maximum degree $2k$.
    
    Now to output the kernel as an EA stream, we make a fourth pass and only output those edges not incident to vertices in $S$. In this procedure we only use $\O(k\log n)$ bits of memory for counting and $X$. The resulting kernel has at most $2k \cdot (k - |S|) = \O(k^2)$ edges.
\end{proof}

Next, we show how to use Theorem~\ref{thm:VCkernelk^2edgesEAAL} to produce a kernel of even smaller size, using only $\O(k\log n)$ bits of memory. This requires Theorem~\ref{thm:VCkernelk^2edgesEAAL} to convert the original graph stream into the kernel input for the next theorem, which only increases the number of passes by a factor 2 or 4 (we have to apply Theorem~\ref{thm:VCkernelk^2edgesEAAL} every time the other procedure uses a pass).

Chen et al.~\cite{CHEN2001_VC} show a way to convert the kernel of Buss and Goldsmith into a $2k$-vertex kernel for \textsc{Vertex Cover [$k$]}, using the NT-Theorem by Nemhauser and Trotter~\cite{NT-theorem}. We will adapt this method in the streaming setting, and give a concise description of this procedure below. The following theorem, as formulated this way by Chen et al.~\cite{CHEN2001_VC}, is due to Nemhauser and Trotter~\cite{NT-theorem} and Bar-Yehuda and Even~\cite{BarYehudaEven}.

	\begin{proposition}[NT-Theorem]\label{prop:NT-theorem}
    	There is an $\O(\sqrt{n}m)$ time algorithm that, given a graph $G$ of $n$ vertices and $m$ edges, constructs two disjoint subsets $C_0$ and $V_0$ of vertices in $G$ such that
    	\begin{enumerate}
        	\item[(1)] The union of any minimum vertex cover of $G[V_0]$ and $C_0$ forms a minimum vertex cover for $G$.
     		\item[(2)] Any minimum vertex cover of $G[V_0]$ contains at least $\lvert V_0 \rvert/2$ vertices.
    	\end{enumerate}
	\end{proposition}

	The proof of the NT-Theorem by Bar-Yehuda and Even~\cite{BarYehudaEven} shows us how to do find $V_0$ and $C_0$ for any arbitrary graph $G$. One creates a bipartite graph $B$ from $G$ by making two copies of all vertices $V, V'$, and an edge $(x,y)$ in $G$ translates to the edges $(x,y')$, $(x',y)$ in $B$. One then finds a maximum matching $M$ of $B$ to find a minimum vertex cover of $B$, as described by Bondy and Murty~\cite[Page~74, Theorem~5.3]{BondyMurtyGraphTheory}. Let us shortly go over what it entails. If our bipartite graph has vertex sets $V, V'$ and a maximum matching $M$, then we can find a minimum vertex cover $X$ with $\lvert X\rvert = \lvert M\rvert$ in the following manner. Denote all unmatched vertices in $V$ with $U$, and let $Z \subseteq V \cup V'$ be the set of vertices connected to $U$ with an $M$-alternating path (a path such that edges in $M$ and not in $M$ alternate). Denoting $S = Z \cap V$ and $T = Z \cap V'$, then $X$ is given by $X = (V \setminus S) \cup T$. Now, the set $C_0$ is given by all vertices $v\in G$ for which both $v,v'\in B$ are contained in $X$, and $V_0$ contains the vertices $v\in G$ for which exactly one of $v,v'\in B$ is contained in $X$.

	Chen et al.~\cite{CHEN2001_VC} describe how to use the above procedure to get a smaller kernel from the kernel by Buss and Goldsmith~\cite{BussVCkernel}. Start with $G$ as the kernel by Buss and Goldsmith~\cite{BussVCkernel}, and execute the above procedure to find the sets $C_0$ and $V_0$ in $G$. Then, the kernel is given by $G' = G[V_0]$ with new parameter value $k' = k_1 - \lvert C_0 \rvert$, where $k_1$ is the parameter value of the kernel by Buss and Goldsmith. Chen et al.~\cite{CHEN2001_VC} show that $G'$ has at most $2k$ vertices and is a kernel with parameter value $k'$.

	We now show how to execute this procedure in the streaming setting, both in the AL and EA models. Note that in the EA model, Theorem~\ref{thm:VCkernelk^2edgesEAAL} yields a kernel which is not exactly the kernel by Buss and Goldsmith, but still has $\O(k^2)$ edges. This property is sufficient for the above procedure to work, even though it is not exactly the kernel by Buss and Goldsmith.

\begin{restatable}{lemma}{BussToBipartite}\label{lemma:BussToBipartite}
	Given a graph $G$ as a stream in model AL or EA, we can produce a stream in the same model corresponding to the Phase 1 bipartite graph $B$ of \cite[Algorithm NT]{BarYehudaEven} using two passes and $\O(\log n)$ bits of memory.
\end{restatable}
\begin{proof}
    Given a graph $G = (V,E)$, Phase 1 of \cite[Algorithm NT]{BarYehudaEven} asks for the bipartite graph $B$ with vertex sets $V, V'$ and edges $E_B$ such that $V' = \{v' \mid v\in V \}$ and $E_B = \{ (x,y') \mid (x,y) \in E\}$. This is essentially two copies of all vertices and each edge in the original graph makes two edges, between the corresponding (original,copy)-pairs.
	
	The process of creating a stream corresponding to $B$ is quite simple: first we use a pass and, for every edge $(x,y)$, we output $(x,y')$, and then we use another pass and, for every edge $(x,y)$, we output $(x',y)$. If the input is an EA stream, then the output must be as well, as no edge is output twice. If the input stream is an AL stream, the output must be an AL stream too, as we are consistent in which copy of the vertex we address. That is, the output AL stream first reveals all vertices in $V$ and then all those in $V'$. All adjacencies of these vertices are present in the stream, as all the adjacencies were present in the input stream.
	
	We can see that this uses two passes and $\O(\log n)$ bits of memory (to remember what pass we are in, and to read an edge). It is trivial that the bipartite graph $B$ is constructed correctly, as for every edge $(x,y)$ we output the edges $(x,y')$ and $(x',y)$.
\end{proof}

Before we continue to find the maximum matching in such a graph $B$ produced by Lemma~\ref{lemma:BussToBipartite}, we need a few observations to restrict the size of the matching we want to find. From the conversion to a $2k$ kernel by Chen et al.~\cite{CHEN2001_VC}, we can conclude that for the sets $C_0$ and $V_0$ of the NT-Theorem it must hold that $\lvert V_0\rvert \leq 2k - 2 \lvert C_0\rvert$ (as this shows the kernel size). But then it must also be that $\lvert V_0 \rvert + \lvert C_0\rvert \leq 2k$, and $V_0$ and $C_0$ together include all vertices in the found minimum vertex cover in $B$. So, the maximum matching $M$ of $B$ we search for has size $\O(k)$.

To find the maximum matching we execute a DFS procedure, which can be done with surprising efficiency in this restricted bipartite setting.

\begin{restatable}{theorem}{MaximumMatchingBipartite}\label{thm:MaximumMatchingBipartite}
	Given a bipartite graph $B$ as an AL stream with $\O(k^2)$ vertices, we can find a maximum matching of size at most $\O(k)$ using $\O(k^2)$ passes and $\O(k\log n)$ bits of memory. For the EA model this can be done in $\O(k^3)$ passes.
\end{restatable}
\begin{proof}
    We first use a pass to find a maximal matching $M$ in the graph. This can be done in a single pass because we can construct a maximal matching in a greedy manner, picking every edge that appears in the stream for which both vertices are unmatched.
	
	Then, we iteratively find an $M$-augmenting path $P$ (a path starting and ending in a unmatched vertices, alternating between edges in and not in $M$), and improve the matching by switching all edges on $P$ (i.e., remove from $M$ the edges on $P$ in $M$, and add to $M$ the edges on $P$ not in $M$). Note that any such $P$ has length $\O(k)$, as otherwise $M$ would exceed size $\O(k)$. It is known that a matching $M$ in a bipartite graph is maximum when there is no $M$-augmenting path \cite{HopcroftKarpMaxMatching}. We can also find an $M$-augmenting path only $\O(k)$ times, as the size of the matching increases by at least 1 for each $M$-augmenting path.
	
	Let us now describe how we find an $M$-augmenting path, given some matching $M$ of size $\O(k)$. We find $M$-augmenting paths by executing a Depth First Search (DFS) from each unmatched vertex. Note that we alternate between traversing edges in $M$ and not in $M$ in this search. In contrary to a normal DFS, we do not save which vertices we visited, as this would cost too much memory. Instead, we mark edges in $M$ as visited, together with the vertex from which we started the search. If an edge $e \in M$ has been visited once in the search tree, there is no need to visit it again, as the search that visited $e$ would have found an $M$-augmenting path containing $e$ if it exists. Let us discuss the exact details on the size of the search tree and recursion.
	
	As any $M$-augmenting path has length at most $\O(k)$, the depth of the search tree is also $\O(k)$. Looking at any vertex, it might have $\O(k^2)$ neighbours in the given bipartite graph. However, only $\O(k)$ of its neighbours can be in $M$. As visiting an unmatched vertex must end the $M$-augmenting path, the search tree size is only increased by visiting matched vertices. Therefore, the search comes down to the following process. From the initial unmatched vertex, we can explore to at most $\O(k)$ vertices (those in the matching) or any unmatched vertex which would end the search. If we explore to a matched vertex, the next step must traverse the edge in the matching to make an $M$-augmenting path, which is deterministic. Then we again can explore to $\O(k)$ matched vertices, or any unmatched vertex which would end the search. This process continues. As we only visit each matched vertex at most once, we can see that the number of vertices the search visits is bounded by $\O(k)$. In each node along the currently active path of the search tree, we can keep a counter with value $\O(k^2)$ (using $\O(\log(k^2)) = \O(\log k)$ bits) to keep track of what edge we consider next. These counters take up $\O(k \cdot \log k) = \O(k \log n)$ space. In any node, if we wish to consider the next edge incident to a vertex $v$ with a counter value $x$, we inspect the $x$-th edge incident to $v$ in the stream. If it turns out we cannot visit that vertex (have already visited it), we can increment the counter and find the next edge to consider in the same pass (as the $(x+1)$-th edge incident to $v$ must be later in the stream than the $x$-th edge incident to $v$). Therefore, finding the next edge to visit in the search only takes a single pass. Notice that we return to nodes in the search tree at most $\O(k)$ times in total, because only visiting matched vertices can result in a `failed' search recursion. So, a search that visits all matched vertices uses $\O(k)$ passes, and this is the maximum number of passes for a single search.
	
	We start our search at most once from each unmatched vertex that has at least one edge (for which we can keep another counter to keep track), which means we do at most $\O(k^2)$ searches. However, for each of these searches we start from different vertices, we still keep saved the set of visited matched vertices. If a search from a vertex visits a matched vertex and does not find an $M$-augmenting path, then neither will a search from a different vertex by visiting that matched vertex again. In particular, this is because the graph is bipartite, because, when we start from an unmatched vertex on e.g. the `right' side, we have to end on an unmatched vertex on the `left' side, while all vertices we visit on the path are matched vertices. So, the current path does not interfere with the ability to successfully find an endpoint, which makes another search visiting the same matched edge have exactly the same result. Having to do $\O(k^2)$ searches would indicate that we need to use $\O(k^2)$ passes, at least one for every search. However, in the AL model, if we consider the $x$-th vertex, and in the pass we use for it, we do no successful visit to a vertex (all adjacencies are matched and already visited), then in the same pass we can consider the $(x+1)$-th vertex, because all edges incident to the $(x+1)$-th vertex in the stream appear later than the edges incident to the $x$-th vertex in the stream. Hence, in the AL model, over all $\O(k^2)$ starting vertices, we only use $\O(k)$ passes, because only (partially) successful searches increase the number of passes, and we can only visit $\O(k)$ vertices in total. In the EA model, we require at least one pass for each vertex we want to start searching from, and so the total number of passes is $\O(k^2)$.
	
	We conclude that with $\O(k)$ passes in the AL model, and $\O(k^2)$ passes in the EA model, and $\O(k \log n)$ bits of memory we can execute a DFS to find an $M$-augmenting path (if it exists).
	
	As mentioned, we can search for an $M$-augmenting path only $\O(k)$ times, as the existence of more $M$-augmenting paths would result in returning NO. Therefore, we can find a maximum matching in $B$ using $\O(k^2)$ passes and $\O(k \log n)$ bits of memory in the AL model. In the EA model, we require $\O(k^3)$ passes to accomplish this.
\end{proof}

Next, we show how to convert such a maximum matching as found by Theorem~\ref{thm:MaximumMatchingBipartite} into a minimum vertex cover for $B$, as asked by \cite[Algorithm NT]{BarYehudaEven}, for which we can use a DFS procedure as in Theorem~\ref{thm:MaximumMatchingBipartite} as a subroutine.

\begin{restatable}{lemma}{MaximumMatchingToVCBipartite}\label{lemma:MaximumMatchingToVCBipartite}
	Given a bipartite graph $B$ as an AL stream and a maximum matching $M$ of size $\O(k)$, we can find a minimum vertex cover $X$ for $B$ with $|X| = |M|$, using $\O(k)$ passes and $\O(k \log n)$ bits of memory. For the EA model, this takes $\O(k^2)$ passes.
\end{restatable}
\begin{proof}
    We adapt a theorem by Bondy and Murty~\cite[Page~74, Theorem~5.3]{BondyMurtyGraphTheory} to the streaming setting to achieve this lemma. Let us repeat again what it entails. If our bipartite graph has vertex sets $V, V'$ and a maximum matching $M$, then we can find a minimum vertex cover $X$ with $|X| = |M|$ in the following manner. Denote all unmatched vertices in $V$ with $U$, and let $Z \subseteq V \cup V'$ be the set of vertices connected to $U$ with an $M$-alternating path (a path such that edges in $M$ and not in $M$ alternate). If $S = Z \cap V$ and $T = Z \cap V'$, then $X$ is given by $X = (V \setminus S) \cup T$.
	
	As $T \subseteq X$ and $|X| = |M|$, we can find and save $T$ by executing a DFS procedure just like in Theorem~\ref{thm:MaximumMatchingBipartite}, without exceeding $\O(k \log n)$ bits of memory. This takes $\O(k)$ passes in the AL model and $\O(k^2)$ passes in the EA model, and $\O(k \log n)$ bits of memory. Also, $V \setminus S$ must only contain matched vertices, as $U \subseteq S$. Therefore, in the same DFS procedure to find $T$, we can also save for every matched vertex in $V$ if it is reachable through an $M$-alternating path. Then $V \setminus S$ is simply given by all matched vertices in $M$ for which we did not save that they were reachable. We conclude that we can find $X$, the minimum vertex cover such that $|X| = |M|$, in $\O(k)$ passes (AL model) or $\O(k^2)$ passes (EA model) and $\O(k \log n)$ bits of memory.
\end{proof}

The final result is as follows, which consists of putting the original stream through each step for every time we require a pass, i.e. the number of passes of each of the parts of this theorem combine in a multiplicative fashion.

\begin{restatable}{theorem}{FinalKernelVCStreaming}\label{thm:2kKernelVCStreaming}
	Given a graph $G$ as an AL stream, we can produce a kernel of size $2k$ for the \textsc{Vertex Cover [$k$]} problem using $\O(k^2)$ passes and $\O(k \log n)$ bits of memory. In the EA model, this procedure takes $\O(k^3)$ passes.
\end{restatable}
\begin{proof}
    We execute Theorem~\ref{thm:MaximumMatchingBipartite} on the stream produced by applying Theorem~\ref{thm:VCkernelk^2edgesEAAL} and then Lemma~\ref{lemma:BussToBipartite} on the input stream (we have to apply these transformations every time that we require a pass). Notice that these applications increase the number of passes by a constant factor. On the result of Theorem~\ref{thm:MaximumMatchingBipartite} we apply Lemma~\ref{lemma:MaximumMatchingToVCBipartite} to obtain a minimum vertex cover for the specific bipartite graph $B$. Now, $C_0$ contains the vertices $v$ for which both $v,v' \in B$ are contained in the minimum vertex cover of $B$, and $V_0$ contains the vertices $v$ where either $v,v' \in B$ is contained in the minimum vertex cover of $B$, but not both. Finding $C_0$ and $V_0$ from $B$ and its minimum vertex cover requires no passes over the stream, as they are simply given by analysing the minimum vertex cover of $B$. These sets $C_0$ and $V_0$ are exactly the sets in the NT-Theorem (Proposition~\ref{prop:NT-theorem}). The kernel by Chen et al.~\cite{CHEN2001_VC} is given by $G' = G[V_0]$, which we can find with a pass (we can output the kernel as a stream), and parameter $k' = k_1 - |C_0|$, where $k_1$ is the parameter after application of Theorem~\ref{thm:VCkernelk^2edgesEAAL}. All in all, this process takes $\O(k^2)$ (AL) or $\O(k^3)$ (EA) passes and $\O(k\log n)$ bits of memory.
\end{proof}


\section{Conclusion}
We studied the complexity of \textsc{Diameter} and \textsc{Connectivity} in the streaming model, from a parameterized point of view. In particular, we considered the viewpoint of an $H$-free modulator, showing that a vertex cover or a modulator to the disjoint union of $\ell$ cliques effectively forms the frontier of memory- and pass-efficient streaming algorithms. Both problems remain hard for almost all other $H$-free modulators of constant size (often even of size~$0$). We believe that this forms an interesting starting point for further investigations into which other graph classes or parameters might be useful when computing \textsc{Diameter} and \textsc{Connectivity} in the streaming model.

On the basis of our work, we propose four concrete open questions:
\begin{itemize}
    \item What is the streaming complexity of computing \textsc{Distance to $\ell$ Cliques}? On the converse of \textsc{Vertex Cover [$k$]}, we are not aware of any algorithms to compute this parameter, even though it is helpful in computing \textsc{Diameter} and \textsc{Connectivity}.
    \item Are there algorithms or lower bounds for \textsc{Diameter} or \textsc{Connectivity} in the AL model for interval graphs? 
    \item Assuming isolated vertices are allowed in the graph, can we solve \textsc{Connectivity} in the AL model on split graphs using $O(\log n)$ bits of memory?
    \item Is there a streaming algorithm for \textsc{Vertex Cover [$k$]} using $\O(\mathrm{poly}(k))$ passes and $\O(\mathrm{poly}(k, \log n))$ bits of memory, or can it be shown that one cannot exist? This result would be relevant in combination with our kernel.
\end{itemize}

\bibliographystyle{abbrv}
\bibliography{bib}

\end{document}